\documentclass[a4paper,12]{article}

\usepackage{subfig}

\usepackage{graphicx}
\usepackage{epsf}

\usepackage{amsmath} 
\usepackage{amssymb}

\usepackage{amsthm}
\usepackage{widebar}
\usepackage{natbib}
\usepackage{mathtools}  
\usepackage{enumerate}

\theoremstyle{definition}
\newtheorem{theorem}{Theorem}
\newtheorem{example}{Example}[section]

\newtheorem{axiom}{Axiom}

\DeclareMathOperator*{\argmin}{argmin}
\DeclareMathOperator*{\argmax}{argmax}

\usepackage{cleveref}

\usepackage{xcolor}
\usepackage{todonotes}
\definecolor{nathaniel_color}{rgb}{0.7,1,0.7}
\definecolor{daniel_color}{rgb}{0.7,0.7,1}
\definecolor{nihat_color}{rgb}{1,0.7,0.4}

\newcommand{\antichainfullf}{\mathcal}

\newcommand{\simpcompfullf}{\mathcal}

\newcommand{\simpcomp}{\mathfrak{S}}

\newcommand{\coalition}{\simpcomp}
\newcommand{\someothercoalition}{\mathfrak{U}}
\newcommand{\yetanothercoalition}{\mathfrak{T}}
\newcommand{\allsubsets}{\mathfrak{N}}
\newcommand{\allsimpcomp}{\mathcal{D}}
\newcommand{\allcoalitions}{\allsimpcomp}
\newcommand{\allplayers}{\allsubsets}
\newcommand{\cooperativegame}{\Upsilon}
\newcommand{\allrankings}{\mathcal{R}}

\newcommand{\kl}{{D_\text{KL}}} % this symbol is used as a shorthand for the KL divergence from one lattice node to another. (It was 'I' before.)

\title{Information Decomposition Based on Cooperative Game Theory}  

\author{Nihat Ay${}^{1,2,3}$, Daniel Polani$^{4}$, Nathaniel Virgo$^{1,5}$}

\begin{document}

\maketitle

\begin{center}
${}^{1}$Max Planck Institute for Mathematics in the Sciences, Leipzig, Germany \\
${}^{2}$University of Leipzig, Leipzig, Germany \\
${}^{3}$Santa Fe Institute, Santa Fe, NM, USA \\
${}^{4}$%Adaptive Systems Research Group, Department of Computer Science, 
University of Hertfordshire, UK
  \\
${}^{5}$Earth-Life Science Institute (ELSI), Tokyo, Japan \\
\end{center}

\begin{abstract} 
We offer a new approach to the \emph{information decomposition} problem in information theory: given a `target' random variable co-distributed with multiple `source' variables, how can we decompose the mutual information into a sum of non-negative terms that quantify the contributions of each random variable, not only individually but also in combination? We derive our composition from cooperative game theory. It can be seen as assigning a ``fair share'' of the mutual information to each combination of the source variables. Our decomposition is based on a different lattice from the usual `partial information decomposition' (PID) approach, and as a consequence our decomposition has a smaller number of terms: it has analogs of the synergy and unique information terms, but lacks terms corresponding to redundancy. Because of this, it is able to obey equivalents of the axioms known as `local positivity' and `identity', which cannot be simultaneously satisfied by a PID measure.
\\
 

%{
%\small
%\tableofcontents
%}
\end{abstract}

%     ************************************************** SECTION **************************************************

\section{Introduction}

We are interested in understanding the flow of information in dynamical systems. Several tools for this have been developed over recent decades. These include  the
transfer entropy
\citep{schreiber00:_measur_infor_trans,lizier14:_jidt}, which assumes a
dynamical systems framework and that the system consists of identifiable
components which can be tracked through time, and the causal
information flow \citep{ay2008information}, which is  defined in terms
of general causal Bayesian networks and the interventional calculus~\citep{pearl2009causality}.
These techniques provide ways to
quantify the amount of influence that one part of the network
has on another. 

However, there is a growing awareness
that it would be useful to quantify causal influences in a more
fine-grained way than offered by current techniques.
Even in one of the simplest cases, in which multiple random variables exhibit a
causal influence on a single `target' variable, it would be desirable to have more detailed understanding in information theoretic terms,
not only of how each variable affects the target \emph{individually}, but of how \emph{multiple} causes interact in bringing
about their effects.

Of the existing approaches to this question, perhaps the best known is the
 Partial Information Decomposition (PID) framework, due to
\citet{williams2010nonnegative}. 
The PID framework proposes that the mutual information between several `source' variables 
and a single `target' can be decomposed into a sum of several terms. In the case of two sources,
these terms are $(i)$ the information that the two sources provide redundantly about the target
(known as redundant information, shared information or common information); $(ii)$ the information provided
uniquely by each of the two sources, and $(iii)$ the synergistic or complementary information, which can only
be obtained by knowing both of the sources simultaneously.

However,
the axioms proposed by Williams and Beer do not completely determine
these quantities. As a result, many PID measures have been proposed in
the literature, each satisfying different additional properties beyond
the ones given by Williams and Beer. Several approaches have been proposed.
{
Among these are several that are based on information geometry 
\citep{harder13:_bivar, bertschinger2014quantifying, perrone2016hierarchical, olbrich2015information,
griffith2014quantifying, James:2019jn}, which we build upon here.
}

Generalizing towards the case of three or more input variables has
turned out to be more problematic under the PID framework. One of the
most intuitive additional axioms proposed is known as the identity axiom, proposed by
\cite{harder13:_bivar}, but it was shown by \cite{rauh2014reconsidering} that no
measure can exist that obeys both Williams and Beer's axioms
(including ``local positivity'') and the identity axiom. 
{
Because of this, there are a number of
proposed PID measures that relax either the identity axiom or the local positivity axiom
of Williams and Beer, or both. Such approaches include \citep{ince2017partial, finn2018pointwise, kolchinsky2019novel}.
Another promising class of approaches involve changing to a slightly different perspective, for example, by considering the
full joint distribution between multiple random variables, rather than singling out a single variable as the target 
\citep{rosas2016understanding, james2017multivariate}. In the present paper, we present a different decomposition of the
mutual information between a set of sources and a target. Our decomposition obeys analogs of both the local positivity and
identity axioms, but it has a smaller number of terms than the partial information decomposition.
}

A related, but different, approach to multivariate information can be
found in a family of measures that attempt quantify the complexity of
a set of random variables, often also divided into input and output
variables. These include Amari's \emph{hierarchical decomposition}
\citep{amari2001information}, as well as Ay's measure of complexity \citep{ay2015information}, and several measures that have arisen in the context of Integrated Information Theory (IIT), such as \citep{oizumi2016unified}. This family of measures is reviewed in \citep{amari2016geometry}, which describes their relationships in terms of information geometry.

In this paper we are interested in a similar setup to the PID framework, in which several   random variables, $X_1, X_2, \dots X_n$, which we term 
\emph{input random variables}, exhibit causal influences on a \emph{target} random variable $Y$.   
This results in a joint distribution between  $X_1, X_2, \dots X_n$ and $Y$.
The mutual information
and conditional mutual information can be used to quantify the
influence of each cause individually, as well as the conditional
influence that one input variable has, once the value of another has been
taken into account. However, in general it would be desirable to decompose the
relationships between the causal influences more finely than the
traditional conditional mutual information makes possible.

The promise of the PID approach was that it would offer a ready-made
or at least preferred solution to this question. A PID measure would
have allowed us to quantify not only the overall influence of $X_1$
upon $Y$, but also the extent to which it has a \emph{unique} causal
influence, which could be interpreted as distinct from that of the
other causes; additionally, it would also allow synergistic or
redundant causal influences to be quantified. This could be done
simply by applying the PID to the joint distribution
$X_1 X_2 \dots X_n Y$. This is, broadly speaking, the approach taken
by \cite{lizier2014framework}. However, 
the lack of a non-negative
PID measure for three or more input variables makes it difficult to
interpret the decomposition in the case of three or more causes.

Here we propose a different approach, with slightly more modest goals
than PID, in that we do not attempt to quantify redundancy. Instead, we
provide a decomposition of the mutual information $I(X_1, X_2, \dots,
X_n ;Y)$ into a sum of terms corresponding to every possible subset of
the causes (e.g.\ $\{X_1\}$, $\{X_1, X_2\}$, etc.). These terms resemble the unique information and synergy terms in the PID framework.
We show that the
terms represent, in a well-defined sense, a uniquely ``fair
apportionment'' of the total mutual information into the contribution
provided by each subset of sources. A set such as $\{X_1, X_2\}$ will
make a contribution of zero if it provides no new information beyond
that which is already provided by its subsets. How to achieve this  will be made precise below. In this sense our measure plays a similar role to that of synergy in the PID lattice. However, our measure does not attempt to quantify redundancy, and as such it is not a solution to the PID problem. Because of this, we are able to give a non-negative decomposition for an arbitrary number of inputs, which obeys an analog of the identity axiom for PID measures.

To state our problem more precisely: we consider random variables $X_1,\dots, X_n$, the {\em input
  variables\/}, and $Y$, the {\em output variable\/}. We restrict ourselves to the case where these variables have finite
state sets $\mathsf{X}_i$, $i = 1,\dots,n$, and $\mathsf{Y}$, but we expect our measure to generalise well to cases such
as Gaussian models in which the state spaces are continuous.

We write $V$ for the set of all input variables. Our
goal is a decomposition of the mutual information, $I(X_1, \dots, X_n;
Y)$ into a sum of terms corresponding to every {}
subset of $\{X_1, \dots, X_n\}$. We refer to a set of input variables as
a {\em predictor\/}. In other words, we will write
\[
	I(X_1,\dots, X_n;Y) = \sum_{A\in 2^V} I_A(X_1,\dots, X_n;Y)\;,
\]
where we write $2^V$ for the power set of $V$.
 For a given predictor $A$, the term
$I_A(X_1,\dots, X_n;Y)$ shall indicate the proportion of the total mutual
information that is contributed by $A$, beyond what is already
provided by its subsets. How to do this will be made precise below. We call
$I_A(X_1,\dots, X_n;Y)$ the \emph{information contribution} of $A$ to
$Y$.

To construct our measure of information contribution we proceed in two
steps. We begin by defining the mutual information provided by certain
\emph{sets} of predictors, i.e.\ sets of sets of input variables. We
do this via a sublattice of the lattice of probability distributions
that \cite{James:2019jn} termed the ``constraint lattice.'' The same
lattice has appeared in the literature previously, within the topic of
reconstructability analysis \citep{Zwick:2004co}. Having established
the information contribution of each set of predictors, we then assign
a contribution to each individual predictor by a method that involves 
summing over the maximal chains of the constraint lattice. 

We then show that this procedure of summing over maximal chains
can be derived using cooperative game theory. We can conceptualise our measure in terms of a
cooperative game, in which each set of predictors is thought of as a coalition of players.
Each coalition is assigned a score corresponding to the information they
jointly provide about the target. Our measure can then be derived via a known generalisation
of the Shapley value due to {\cite{Faigle1992}}, which assigns a score to each individual player (i.e. predictor)
based on its average performance among all the coalitions in which it takes part, while respecting additional precedence constraints.

Since our measure is based on the constraint lattice, we review this concept in depth in \cref{constraint_lattice.sec}. We approach the constraint lattice from the perspective of information geometry and state its relationship to known results in that field. In \cref{input_lattice.sec} we consider a sublattice of the constraint lattice which we term the \emph{input lattice}, which allows us to define a quantity corresponding to the information that a set of predictors provides about the target. From this we derive our measure by summing over the maximal chains of the input lattice. After proving some properties of our information contribution measure and giving some examples (\cref{properties,examples.sec}), we then make the connection to cooperative game theory in \cref{cooperative_game_theory}, proving that our measure is equivalent to the generalised Shapley value of {\cite{Faigle1992}}.

\section{Background: the constraint lattice}
\label{constraint_lattice.sec}

We begin by defining the so-called ``constraint lattice'' of \cite{James:2019jn}, which has also been defined previously in the context of reconstructability analysis \citep{Zwick:2004co}. This  section serves to summarise previous work and to establish notation for the following sections.

\subsection{Lattices of simplicial complexes}

Suppose we have a set $W$ of co-distributed random variables, $W = \{Z_1, Z_2, \dots, Z_m\}$. Subsets of $W$ may also be considered as random variables. For example, $\{Z_1, Z_2\}$, which we also write $Z_1Z_2$, can be thought of as a random variable whose sample space is the Cartesian product of the sample spaces of $Z_1$ and $Z_2$. The constraint lattice is defined in terms of members of the power set $2^W$.%

For reasons that will be explained below, %in Section~\ref{constraints}
 we want to put some restrictions on which members of $2^W$ are permitted as elements of the lattice. This may be done
in terms of of two different concepts, \emph{antichains} or \emph{simplicial complexes}. It is standard in the literature to define the constraint lattice in terms of antichains. However, in making the connection to cooperative game theory it will be more convenient to talk in terms of simplicial complexes instead. For this reason, we define both terms here, but define the constraint lattice in terms of simplicial complexes.

A set of %non-empty 
sets $\antichainfullf{S}$
 is called an antichain if for every $A \in \antichainfullf{S}$, no
 subset $B\subset A$ 
 is a member of~$\antichainfullf{S}$, i.e. 
\begin{equation} \label{antichain}
     A \in \antichainfullf{S}, \;\; B \subset A \quad \Rightarrow \quad B \notin \antichainfullf{S}\,.
\end{equation}
Similarly, a %non-empty 
set of sets $\simpcompfullf{S}$ is called a simplicial complex if for every $A \in \simpcompfullf{S},$ every subset $B\subset A$ is a member of $\simpcompfullf{S},$ 
\begin{equation} \label{simplcomp}
     A \in \simpcompfullf{S}, \;\; B \subset A \quad \Rightarrow \quad B \in \simpcompfullf{S}\,.
\end{equation}
Both of these concepts can be defined more generally in terms of arbitrary partial orders, but here we need only their definitions in terms of sets. In some contexts the empty set is considered a simplicial complex, but for most of this paper we consider only non-empty simplicial complexes.

We note that one can convert an antichain $\antichainfullf{S}$ into a simplicial complex by adding every subset of every member of $\antichainfullf{S}$, and one can restore the antichain by removing every element that is a subset of some other element. This gives a one-to-one correspondence between antichains and simplicial complexes, which allows us to use the two concepts somewhat interchangeably.

We define the
constraint lattice in terms of simplicial complexes whose members
\emph{cover}~$W$, meaning those simplicial complexes
$\simpcompfullf{S}$ for which each element of $W$ appears at least once in
one
of the members of $\simpcompfullf{S}$. That is, $\simpcompfullf{S}$
covers $W$ if $\bigcup_{A\in\simpcompfullf{S}}
A=W$. 
 Such a simplicial complex is termed a \emph{simplicial complex cover}
 of $W$.
 
The following partial order may be defined on simplicial complexes:
\begin{equation}
\label{simp_comp_partial_order} 
\simpcompfullf{S}\le \simpcompfullf{S}' \quad\text{if and only if}\quad \forall A\in \simpcompfullf{S},\, \exists B\in \simpcompfullf{S}'{:}\, A\subseteq B\;.
\end{equation}
The constraint lattice is composed of the simplicial complexes that cover $W,$ with this partial order. 

\begin{figure}
 \centering
 \includegraphics[width=0.5\textwidth]{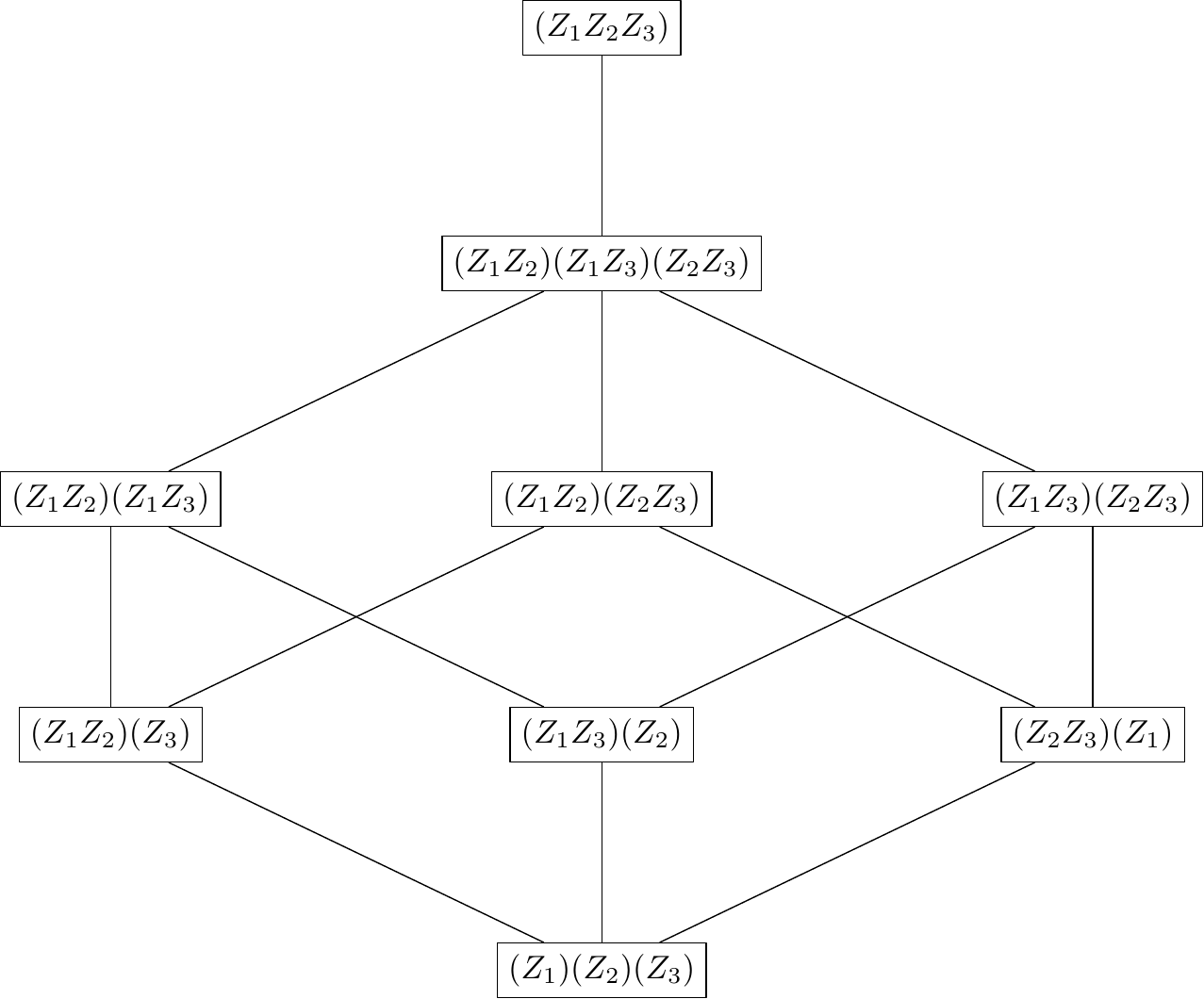}
 \caption{
 	The Hasse diagram for the constraint lattice, as defined by \citep{Zwick:2004co, James:2019jn}, for three random variables, $W = \{Z_1, Z_2, Z_3\}$.
 }
 \label{lattice_z.fig}
\end{figure}

The resulting lattice is illustrated in \cref{lattice_z.fig}. In the figures and elsewhere, we use the following shortcut notation for simplicial complexes: we take the corresponding antichain, write its elements as lists surrounded by parentheses, and concatenate them.
For example, the notation $(Z_1Z_2)(Z_3)$ refers to the simplicial complex $\{\{Z_1, Z_2\}, \{Z_1\}, \{Z_2\}, \{Z_3\}, \emptyset\}$.

It is helpful to introduce some definitions from lattice theory. We write
$\simpcompfullf{S} < \simpcompfullf{T}$ if $\simpcompfullf{S} \le \simpcompfullf{T}$
and $\simpcompfullf{S} \ne \simpcompfullf{T}$. We say that a lattice element
$\simpcompfullf{T}$ \emph{covers} an element $\simpcompfullf{S}$, written
$\simpcompfullf{S} \prec \simpcompfullf{T}$, if $\simpcompfullf{S} <
\simpcompfullf{T}$ and there exists no element $\simpcompfullf{U}$ such that
$\simpcompfullf{S} < \simpcompfullf{U} < \simpcompfullf{T}$.  A sequence of elements $\simpcompfullf{S}_1, \dots, \simpcompfullf{S}_k$ is called a \emph{chain} if $\simpcompfullf{S}_1 < \simpcompfullf{S}_2 < \dots < \simpcompfullf{S}_k$. If we have in addition that  $\simpcompfullf{S}_1 \prec \simpcompfullf{S}_2 \prec \dots \prec \simpcompfullf{S}_k$, then it is called a \emph{maximal chain}.

In \cref{lattice_z.fig}, the relationship $\simpcompfullf{S} \prec \simpcompfullf{T}$ is indicated by drawing $\simpcompfullf{T}$ above $\simpcompfullf{S}$ and connecting the elements with an edge. The resulting graph is called the Hasse diagram of the lattice. The maximal chains are the directed paths from the bottom node in \cref{lattice_z.fig} to the top node.

\subsection{Constraints and split distributions}
\label{constraints}

Let $p = p(Z_1, \dots, Z_m)$ be the joint probability distribution of
the members of $W$. We call this the \emph{true distribution}.
Following \citep{James:2019jn} and \citep{Zwick:2004co}, we now wish
to associate with each simplicial complex cover $\simpcompfullf{S}$ of $W$
 a joint distribution $p_\mathcal{S} =
p_\mathcal{S}(Z_1, \dots, Z_m)$. In the spirit of
\citep{ay2015information,oizumi2016unified,amari2016geometry} we term these \emph{split distributions}. 
{Each split distribution captures only some of the correlations present in the true distribution, and we can think of the remaining correlations as being split apart, or forced to be as small as possible.}

Specifically, each split distribution $p_\simpcompfullf{S}$ is constructed so that it captures the correlations associated with the members of $\simpcompfullf{S}$, in the sense that $p_\simpcompfullf{S}(A) = p(A)$, for every $A\in \simpcompfullf{S}$. This defines a family of distributions, and from this family we choose the one with the maximum entropy. Intuitively, the maximum entropy distribution is the least correlated one in the family, so it excludes any additional correlations aside from those specified by $\simpcompfullf{S}$.

In the remainder of this section, we define the split distributions more rigorously, alongside some related objects, and we point out an important property, which follows from the so-called {Pythagorean theorem} of information geometry. This section is largely a review of previous work, and makes a connection between the constraint lattice of \citep{James:2019jn, Zwick:2004co} and the language of information geometry \citep[chapter 2]{ay2017book}.

Let $\Delta$ be the set of all joint probability distributions of the
random variables in $W$. For a simplicial complex cover
$\simpcompfullf{S}$, let
\[
M_\simpcompfullf{S} = \{ q \in \Delta : \forall A\in \simpcompfullf{S}, q(A)=p(A) \}\;.
\]
That is, $M_\simpcompfullf{S}$ is the set of all probability
distributions for which the members of $\simpcompfullf{S}$ have the
same marginal distributions as in the true distribution $p$. Note that if the constraint $q(A)=p(A)$ holds for some $A\subseteq W$,
then it will also automatically hold for $B\subseteq A$. This is the
reason for considering simplicial complexes.
$M_\simpcompfullf{S}$ is a mixture family, and we have that $\simpcompfullf{S}\le \simpcompfullf{T} \implies M_\simpcompfullf{S} \supseteq M_\simpcompfullf{T}$.

We can now define the split distribution $p_\simpcompfullf{S}$ as
\begin{equation} \label{maxent}
p_\simpcompfullf{S} = \argmax_{q\in M_\simpcompfullf{S}} H(q).
\end{equation}
Equivalently, we can instead define the split distributions in terms of the Kullback-Leibler divergence, as we will see below. This has the advantage that it is likely to generalise to cases such as Gaussian models in which the state space is not discrete.

There is another interpretation of the split distributions, which is interesting to note. In addition to the mixture family $M_\simpcompfullf{S}$, we can also define an exponential family corresponding to a given node in the constraint lattice. This can be seen as a family of \emph{split models}, i.e.\ probability distributions in which some kinds of correlation are forced to be absent. The split distribution $p_\simpcompfullf{S}$ can be seen as the closest member of this exponential family to the true distribution.

To see this, we define the exponential family
\begin{equation}
\label{exponential_family}
  E_\mathcal{S} =  \left\{ q \in \Delta : q(z_1, \dots, z_m) = \prod_{A\in\mathcal{S}} \mu_A(z_1, \dots z_m), \text{for some set of functions $\mu_A$} \right\}\;,
\end{equation}
where the functions $\mu_A$ have the additional requirements that $\mu_A(z_1, \dots, z_m)$ depends only on $z_i$ for $Z_i\in A$, and  $\mu_A(z_1, \dots, z_m)>0$.
 $E_\mathcal{S}$ is an exponential family, and we have $\simpcompfullf{S}<\simpcompfullf{T} \implies E_\simpcompfullf{S} \subseteq E_\simpcompfullf{T}$.

Finally, we let $\widebar{E}_\mathcal{S}$ be the topological {closure} of the
set ${E}_\mathcal{S}$, meaning that $\widebar{E}_\mathcal{S}$
contains every member of ${E}_\mathcal{S}$, and in addition also
contains all the limit points of sequences in ${E}_\mathcal{S}$. The difference is that ${E}_\mathcal{S}$ does not
contain distributions with zero-probability outcomes, whereas the
closure $\widebar{E}_\mathcal{S}$ does.

Note that we cannot obtain $\widebar{E}_\simpcompfullf{S}$ by simply relaxing the condition that $\mu_A(z_1, \dots, z_m)>0$. This is because although every member of $E_\simpcompfullf{S}$ must factorise according to \cref{exponential_family}, the limit points on the boundary of the simplex can fail to factorise in the same way. An example of this is given by \citep[Example 3.10]{lauritzen1996graphical}. These limit points must be included in order to make sure the split distribution is always defined.

It is a known result in information geometry \citep[Theorem 2.8]{ay2017book}
that for any $\simpcompfullf{S}$, the sets $M_\simpcompfullf{S}$ and
$\widebar{E}_\simpcompfullf{S}$ intersect at a single point. {In fact this
point is the split distribution $p_\simpcompfullf{S}$.}
With the Kullback-Leibler divergence
\[
D_\text{KL}(q\|p) = \sum_{z_1,\dots,z_m} q(z_1,\dots,z_m)\log \frac{ q(z_1,\dots,z_m)}{ p(z_1,\dots,z_m)}\; ,
\]
we can equivalently characterise $p_\simpcompfullf{S}$ by  
\begin{equation}
\label{i-projection}
 p_\simpcompfullf{S} = \argmin_{q \in M_\mathcal{S}} D_\text{KL} (q\| u)\;,
\end{equation}
where $u$ denotes the uniform distribution. This directly follows from (\ref{maxent}). 
A further equivalent characterisation of $p_\simpcompfullf{S}$ is given by 
\begin{equation}
 \label{ri-projection}
 p_\simpcompfullf{S} = \argmin_{q \in E_\simpcompfullf{S}} D_\text{KL} (p\| q)\;.
\end{equation}
In the terminology of information geometry, \cref{i-projection} is an I-projection (information projection) and \cref{ri-projection} is an rI-projection (reverse I-projection). 
The classical theory of these information projections has been 
greatly extended by \cite{csiszarmatus2003, csiszarmatus2004}. 

We also have the so-called \emph{Pythagorean theorem} of information geometry \citep{amari2007methods}, which in our notation says that for simplicial complex covers  $\simpcompfullf{S}<\simpcompfullf{T}<\simpcompfullf{U}$,
\begin{equation}
D_\text{KL} (p_\simpcompfullf{U}\|p_\simpcompfullf{S}) = D_\text{KL} (p_\simpcompfullf{U}\|p_\simpcompfullf{T}) + D_\text{KL} (p_\simpcompfullf{T}\|p_\simpcompfullf{S})\;. \label{pythagorean}
\end{equation}
\Cref{pythagorean} can be extended to any \emph{chain} of elements in the constraint lattice $\simpcompfullf{S}_1 < \simpcompfullf{S}_2 < \dots < \simpcompfullf{S}_k$, to give\begin{equation}
 D_\text{KL}(p_{\simpcompfullf{S}_k} \| p_{\simpcompfullf{S}_1}) = \sum_{i=2}^k  D_\text{KL}(\simpcompfullf{S}_i \| \simpcompfullf{S}_{i-1})\;.
 \label{full_pythagorean_eqn}
\end{equation}
 This will be crucial in defining our information contribution measure below.

Consider the top node in the constraint lattice, given by $(Z_1,\dots,Z_m)$, which we denote $\top$. We have $p_\top = p$. That is, the split distribution corresponding to $\top$ is equal to the true distribution. 

Since we are considering only simplicial complex covers of $W$, the bottom node of the lattice is given by $(Z_1)\dots(Z_m)$, which we denote $\bot$. We have $p_\bot(z_1, \dots, z_m) = p(z_1) \dots p(z_m)$. That is, its split distribution is given by the product of the marginal distributions for all the members of $W$.

Together with \cref{i-projection}, this allows us to interpret $p_\simpcompfullf{S}$ as the distribution that is \emph{as decorrelated as possible} (i.e.\ closest to the product distribution, in the Kullback-Leibler sense), subject to the constraint that the marginals of the members of $\mathcal{S}$ match those of the true distribution. 
Alternatively, via \cref{ri-projection}, we can see it as the distribution that is as close to the true distribution as possible, subject to the constraint that it lies in the closure of the exponential family $\widebar{E}_\simpcompfullf{S}$.

For a general antichain cover $\simpcompfullf{S}$, the split distribution $p_\simpcompfullf{S}$ may not have an analytical solution, and instead must be found numerically. 
One family of techniques for this is iterative scaling \citep[][chapter 5]{csiszar2004information}, which was used to calculate the examples below. Alternatively, one may solve \cref{i-projection} as a numerical optimisation problem, starting from an element such as $\bot$ with a known split distribution. This yields a convex optimisation problem with linear constraints, but it is not always well conditioned.

Finally, given an antichain cover $\simpcompfullf{S}$, we define $I_\simpcompfullf{S} \coloneqq D_\text{KL}(p_\top \| p_\simpcompfullf{S})$. This can be thought of as the amount of information that is present in the true distribution $p_\top$ but is not present in $p_\simpcompfullf{S}$. Note that due to the Pythagorean relation (\cref{pythagorean}) we have $D_\text{KL}(p_\mathcal{T}\|p_\simpcompfullf{S}) = I_\simpcompfullf{S}-I_\simpcompfullf{T}$, for any antichain covers $\simpcompfullf{S}\le \simpcompfullf{T}$. The quantity $I_\mathcal{S}$ turns out to be a useful generalisation of the mutual information, as shown in the following examples.

\begin{example} Independence.
\label{independence_example}
 Suppose $W=\{Z_1,Z_2\}$, and let $\simpcompfullf{S} = (Z_1)(Z_2)$. Then $\widebar{E}_\simpcompfullf{S}$ is the set of distributions $q$ that can be expressed as a product $q(z_1,z_2) = \mu_{1}(z_1)\mu_2(z_2)$, which we may also write $q(z_1,z_2) = q(z_1)q(z_2)$. So $\widebar{E}_\simpcompfullf{S}$ is the set of distributions for which $Z_1$ and $Z_2$ are independent. We have that $p_\mathcal{S} = p(z_1)p(z_2)$, and consequently, it is straightforward to show that in this example, $D_\text{KL}(p_\top \| p_\simpcompfullf{S}) = I(X_1 ; X_2)$.
\end{example}

\begin{example} 
Conditional independence.
\label{ci_example}
Suppose $W=\{Z_1,Z_2,Z_3\}$, and let $\simpcompfullf{S} =
(Z_1Z_3)(Z_2Z_3)$. Then $\widebar{E}_\simpcompfullf{S}$ is the set of
distributions $q$ that can be expressed as a product $q(z_1,z_2,z_3) =
\mu_{1}(z_1,z_3)\mu_2(z_2,z_3)$. These are the distributions
for which $q(z_1,z_2,z_3) = q(z_3)q(z_1|z_3)q(z_2|z_3)$, i.e.\ for which $Z_1\perp\!\!\!\perp_q Z_2 \mid Z_3$.
So in this case $E_\mathcal{S}$ can be seen as a conditional
independence constraint. It is straightforward to show that that $p_\simpcompfullf{S}(z_1,z_2,z_3) = p(z_3)p(z_1|z_3)p(z_2|z_3)$, and consequently $D_\text{KL}(p_\top \| p_\simpcompfullf{S}) = I(X_1 ; X_2 | X_3)$.
\end{example}

\begin{example}
Amari's triplewise information.
\label{triplewise_example}
Suppose $W=\{Z_1,Z_2,Z_3\}$, and let $\simpcompfullf{S} = (Z_1Z_2)(Z_1Z_3)(Z_2Z_3)$. Then $E_\simpcompfullf{S}$ is the set of distributions $q$ that can be expressed as a product $q(z_1,z_2,z_3) = \mu_1(z_1,z_2)\mu_2(z_1,z_3)\mu_3(z_2,z_3)$. 
Unlike the previous two examples, there is no analytic expression for
$\mu_1,$ $\mu_2$ and $\mu_3$ in terms of the probabilities
$q(z_1,z_2,z_3)$. However, \cite{amari2001information} argued that
$\widebar{E}_\simpcompfullf{S}$ can be interpreted as the set of distributions
in which there are no three-way, or ``triplewise'' interactions
between the variables $Z_1,$ $Z_2$ and $Z_3$, beyond those that are
implied by their pairwise interactions. The split distribution
$p_\simpcompfullf{S}$ can be calculated numerically 
as described above, in order to obtain the quantity
$D_\text{KL}(p_\top \| p_\simpcompfullf{S})$, which quantifies the
amount of information present in the triplewise interactions.  \cite{amari2001information} gives a straightforward generalisation, allowing $n$-way interactions to be quantified, among $n$ or more random variables.
 As an example of triplewise information,
consider the case where $Z_1$ and $Z_2$ are uniformly distributed
binary variables, and $Z_3 = Z_1 \mathop{\textsc{xor}} Z_2$. In this
case, in the split distribution $p_{(Z_1Z_2)(Z_1Z_3)(Z_2Z_3)}$ all
three variables are independent. The split distribution has 8 equally likely outcomes while the true 
distribution has 4 equally likely outcomes, leading to a triplewise information of 1
bit.
\end{example}

\section{The input lattice}
\label{input_lattice.sec}

The constraint lattice is defined in terms of an arbitrary set of random variables $W=\{Z_1, \dots, Z_m\}$. We are interested specifically in the case where $W$ is composed of a set of input variables $X_1,\dots,X_n$ and a target variable $Y$. We write $V$ for the set of input variables, so $W=V\cup \{Y\}$.

We wish to decompose the mutual information $I(X_1,\dots, X_n;Y)$ into a sum of terms $I_A(X_1,\dots, X_n;Y)$, one for each subset $A$ of  the input variables.
To do this, we start by noting that 
$$
I(X_1,\dots, X_n;Y) = D_\text{KL}( p_{(X_1,\dots, X_n, Y)}\| p_{(X_1,\dots, X_n)(Y)} )\;.
$$
Because of this, we can use the constraint lattice to derive decompositions of the mutual information.

Consider the set of lattice elements $\simpcompfullf{S}$ such that
$(X_1,\dots, X_n)(Y) \le \simpcompfullf{S}$. This set forms a
sublattice of the constraint lattice, i.e.\ a lattice under the same
partial order. We call this sublattice the input lattice. The input
lattice is highlighted in red in \cref{lattice.fig}, left.

\begin{figure}
 \centering
 \includegraphics[width=0.5\textwidth]{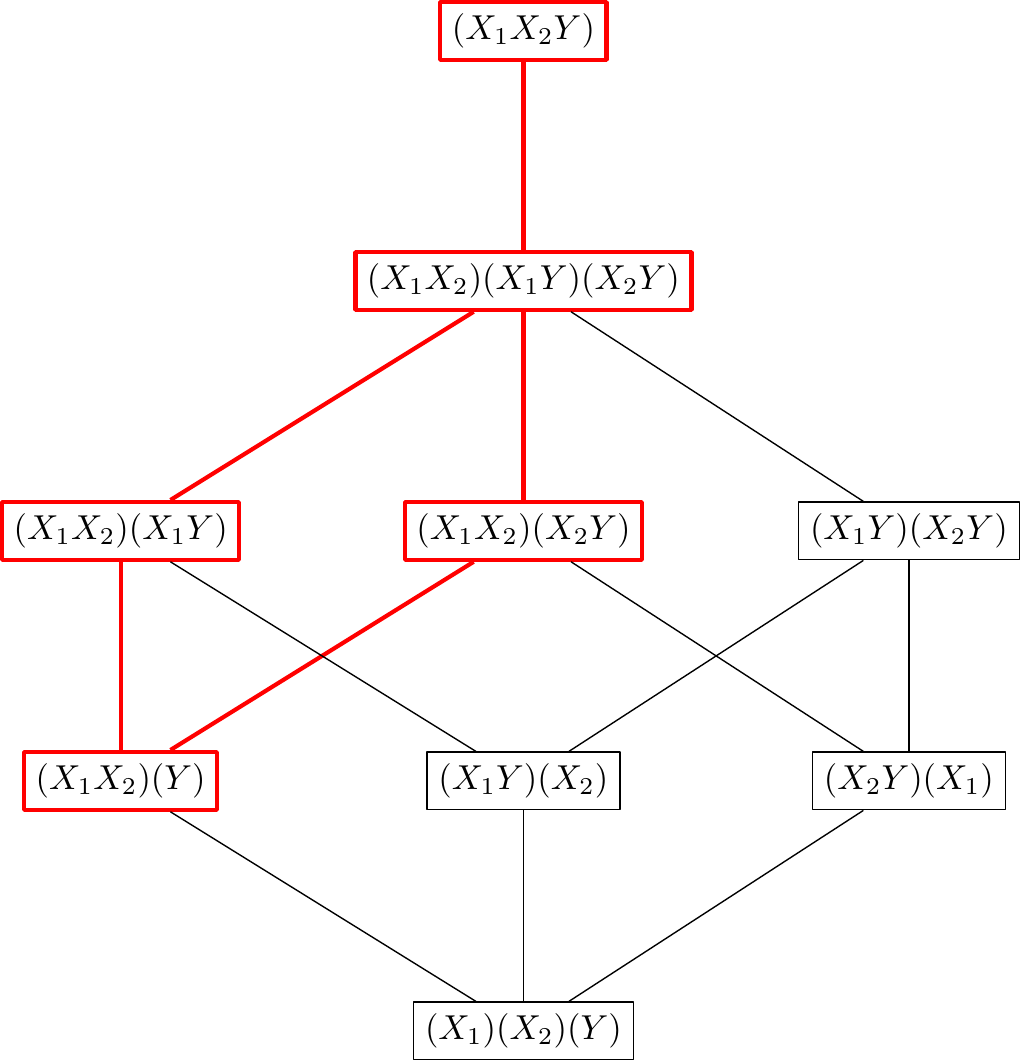}
 \includegraphics[width=0.18\textwidth]{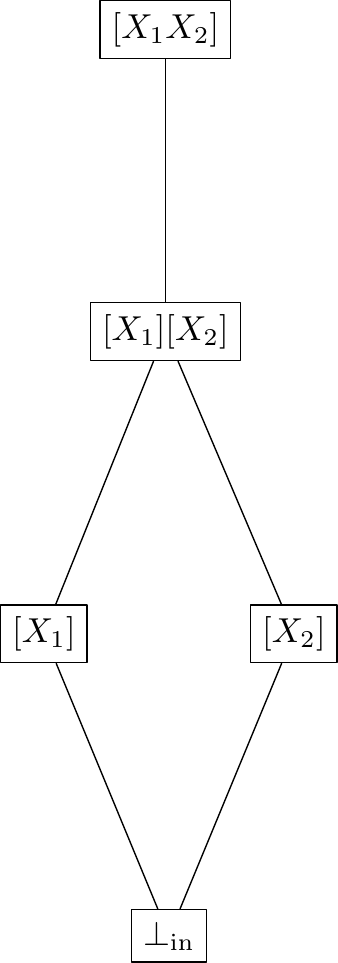}
 \caption{
 	(Left) the Hasse diagram for the constraint lattice for $W=\{X_1,X_2,Y\}$. Highlighted in red bold is the sublattice that we call the input lattice, which provides decompositions of $I(X_1,X_2,\dots, X_n ; Y)$.
	(Right) the input lattice alone, with the nodes labelled using simplicial complexes over the input random variables, rather than all random variables. The square brackets indicate simplicial complexes over the inputs, rather than over all random variables. {The two lattices are related by the mapping $\sigma$, defined in the text.}
 }
 \label{lattice.fig}
\end{figure}

Each element of the input lattice may be associated with a simplicial
complex over the input variables only. That is, a non-empty set
$\simpcomp$ of subsets of $V$, with the condition that every subset of
a member must also be a member. (Unlike the elements of the constraint
lattice, $\simpcomp$ need not cover $V$.) We use a Fraktur font for
simplicial complexes over the input variables only, to distinguish
them from simplicial complex covers of $W$. Their relationship to the
input sublattice can be seen by noting that
{if 
$(X_1\dots X_n)(Y) \le \simpcompfullf{S}$ then $\simpcompfullf{S}$ must include the element
$\{X_1,\dots X_n\}$} 
   and its
subsets. In addition $\simpcompfullf{S}$ contains elements of the form $A\cup\{Y\}$,
where $A$ is a subset of the input variables. These sets of input
variables must by themselves form a simplicial complex, in order for
$\simpcompfullf{S}$ to be a simplicial complex. This is the simplicial complex 
  $\simpcomp$ over the input variables corresponding to $\simpcompfullf{S}$.

Formally, given a simplicial complex $\simpcomp$ over the input variables, {the corresponding member of the constraint lattice is given by}
\[
 \sigma(\simpcomp) = (X_1 \dots X_n) \cup \big\{ A \cup \{Y\}: A \in \simpcomp \big\}\;.
\]
For any $\simpcomp$, we have that $\sigma(\simpcomp)$ is a simplicial complex cover of $W$, and $\sigma(\simpcomp) \ge (X_1\dots X_n)(Y)$. In fact, $\sigma$ is an order-preserving invertible map from the lattice of simplicial complexes $\simpcomp$ over $X$ to the sublattice of simplicial complex covers of $W$ given by $(X_1\dots X_n)(Y) \le \mathcal{S}$. This allows us to think of the elements of the input lattice as corresponding to simplicial complexes over the input variables. 

The mapping is illustrated in \cref{lattice.fig}. When writing simplicial complexes over the input variables explicitly we use square brackets, in order to distinguish them from simplicial complexes over $W$. So for example, the simplicial complex $[X_1][X_2]$ over the input variables corresponds to the simplicial complex $(X_1 X_2)(X_1 Y)(X_2 Y)$ over~$W$. We write the bottom node of the input lattice as $\bot_\text{in}$, which is equal to $\{\emptyset\}$ when considered as a simplicial complex over the input variables, or $(X_1 \dots X_n)(Y)$ when considered as a simplicial complex cover of $W$.

Every chain in this sublattice provides a decomposition of $I(X_1,\dots, X_n;Y)$ into a sum of non-negative terms. An example of such a decomposition is the chain rule for mutual information,
$$
  I(X_1,X_2;Y) = I(X_1;Y) + I(X_2;Y|X_1)\;,
$$
which can be derived by applying the Pythagorean theorem
to the (not
maximal) chain
$$\bot_\text{in} < [X_2] < [X_1X_2]\;.$$
This corresponds to
$$(X_1X_2)(Y) < (X_1X_2)(X_1 Y) < (X_1X_2 Y)$$
when considered as elements of the constraint lattice. {Applying Equation~\ref{pythagorean}, we have
\begin{multline}
   D_\text{KL}(p_{(X_1X_2 Y)} \| p_{(X_1X_2)(Y)}) =   D_\text{KL}(p_{(X_1X_2)(X_1 Y)} \| p_{(X_1X_2)(Y)})  \\+   D_\text{KL}(p_{(X_1X_2 Y)} \| p_{(X_1X_2)(X_1 Y)}),
\end{multline}
which corresponds term-by-term to the chain rule for mutual information.
}
While this chain is not maximal,
considering the maximal chain, however, yields  a more fine-grained
decomposition:
$$\bot_\text{in} < [X_2] < [X_1][X_2] < [X_1X_2]\;,$$
This, in turn, yields an information decomposition with three non-negative terms,
$$
  I(X_1,X_2;Y) = I(X_1;Y) + \big(I(X_2;Y|X_1) - I_3(X_1,X_2,Y)\big) + I_3(X_1,X_2,Y) \;,
$$
where $I_3(X_1,X_2,Y)$ is Amari's triplewise information. (See \cref{triplewise_example} above.)

In this way we can write $I(X_1,\dots, X_n;Y)$ as a sum of non-negative
terms in many different ways. However, these decompositions in general treat the input variables
asymmetrically. The decompositions are ``path-dependent,'' in the
sense that they depend on which particular chain is chosen. In the
next section we  turn these path-dependent decompositions into a
single path-independent one by suitably averaging over the maximal chains.
\section{Defining the information contribution as a sum over chains}

\label{sec:lattice_decomposition}

We now extend the decomposition along individual chains of the input
lattice to a path-independent information decomposition: this
decomposition will define a separate information contribution for each
of the non-empty subsets $A$ of $V$.

In order to do so, consider the set $\Gamma$ of all maximal chains in the input lattice, that is, all directed paths %in \cref{
from 
$\bot_\text{in}$ to $[X_1, \dots, X_n]$.
Consider a maximal chain $\gamma
\in \Gamma$.  For any index $l$ in the chain, the collection
$\gamma(l)$ of subsets forms a simplicial complex and 
for each transition from $\gamma(l)$ to $\gamma(l+1)$ a subset $A$ of $V$ is
added to the simplicial complex $\gamma(l)$ 
until the topmost simplicial
complex which ends up containing all subsets of $V$. In particular, the
chain has  the property that all non-empty
subsets $A$ of $V$ are being added at some point along a chain $\gamma$. 

In particular, this ensures that 
there is exactly one $l_\gamma(A)$ that satisfies the following condition: all simplicial complexes $\gamma(l)$, 
$0 \leq l < l_\gamma(A)$, do not contain $A$, and all simplicial
complexes $\gamma(l)$, $l_\gamma(A) \leq l \leq 2^n - 1$, do contain
$A$; i.e.\ $l_\gamma(A)$ denotes the step in the chain $\gamma$ at
which $A$ is added (note that the empty set $\emptyset$ is necessarily
contained in the first complex of each chain, i.e.\ $l_\gamma(\emptyset) = 0$).

Based on this, we now derive a decomposition of the mutual information
between inputs and output ``aligned'' with respect to a particular
subset $A$ of inputs. For this purpose, consider the set
${\mathcal E}_A$ of all edges $(\mathfrak{S},\mathfrak{S}')$ where
$\mathfrak{S}'$ is obtained from $\mathfrak{S}$ by adding $A$, i.e.\ where
$\mathfrak{S}' = \mathfrak{S}\uplus \{A\}$:
\begin{center}
  \includegraphics[width=0.5\textwidth]{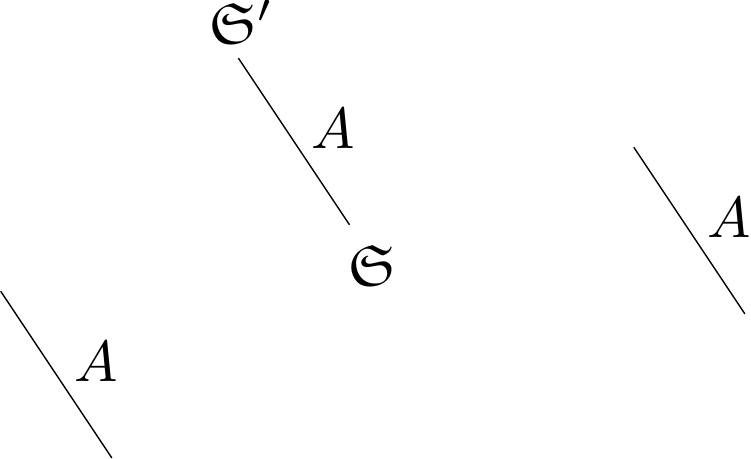}
\end{center}

We furthermore now subdivide the set $\Gamma$ into classes  of
maximal chains, grouped by specific edges $(\mathfrak{S},\mathfrak{S}')$.
Denote by $\Gamma(\mathfrak{S},\mathfrak{S}')$ the set of all
maximal chains $\gamma \in \Gamma$ that contain this particular edge
$(\mathfrak{S},\mathfrak{S}')$:
\begin{center}
  \includegraphics[width=0.55\textwidth]{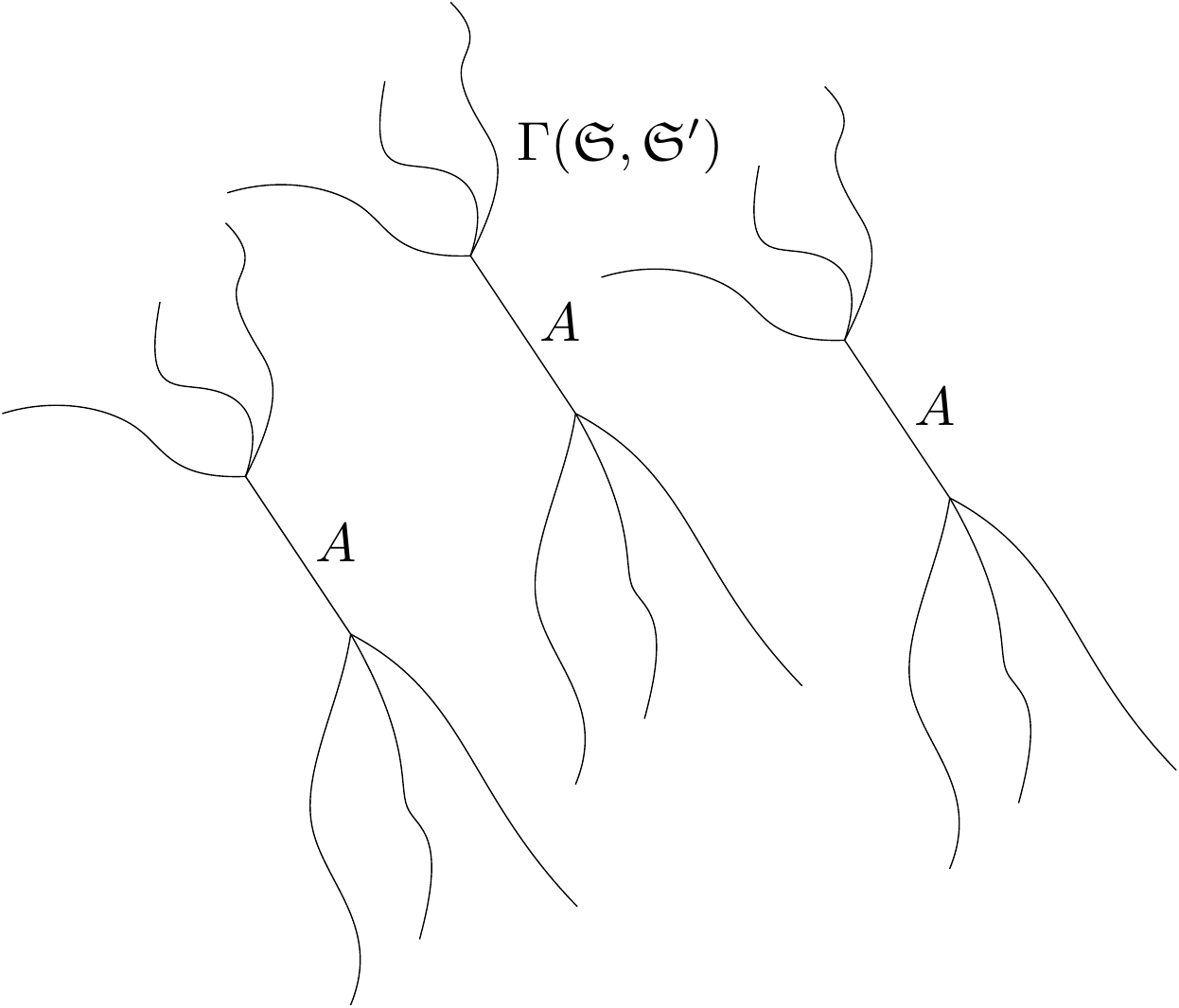}
\end{center}
Then, for any non-empty subset $A$, one has the following partition:
\[
      \Gamma = \biguplus_{(\mathfrak{S},\mathfrak{S}') \in {\mathcal E}_A} \Gamma(\mathfrak{S},\mathfrak{S}') \;.
\]
Every maximal chain is accounted for in this disjoint union, because
for every maximal chain there is exactly one step (edge) at which the set $A$ is added.
This is illustrated in \cref{lattice_3_inputs.fig} for the case of three input variables. 

\begin{figure}
 \centering
 \includegraphics[width=0.3\textwidth]{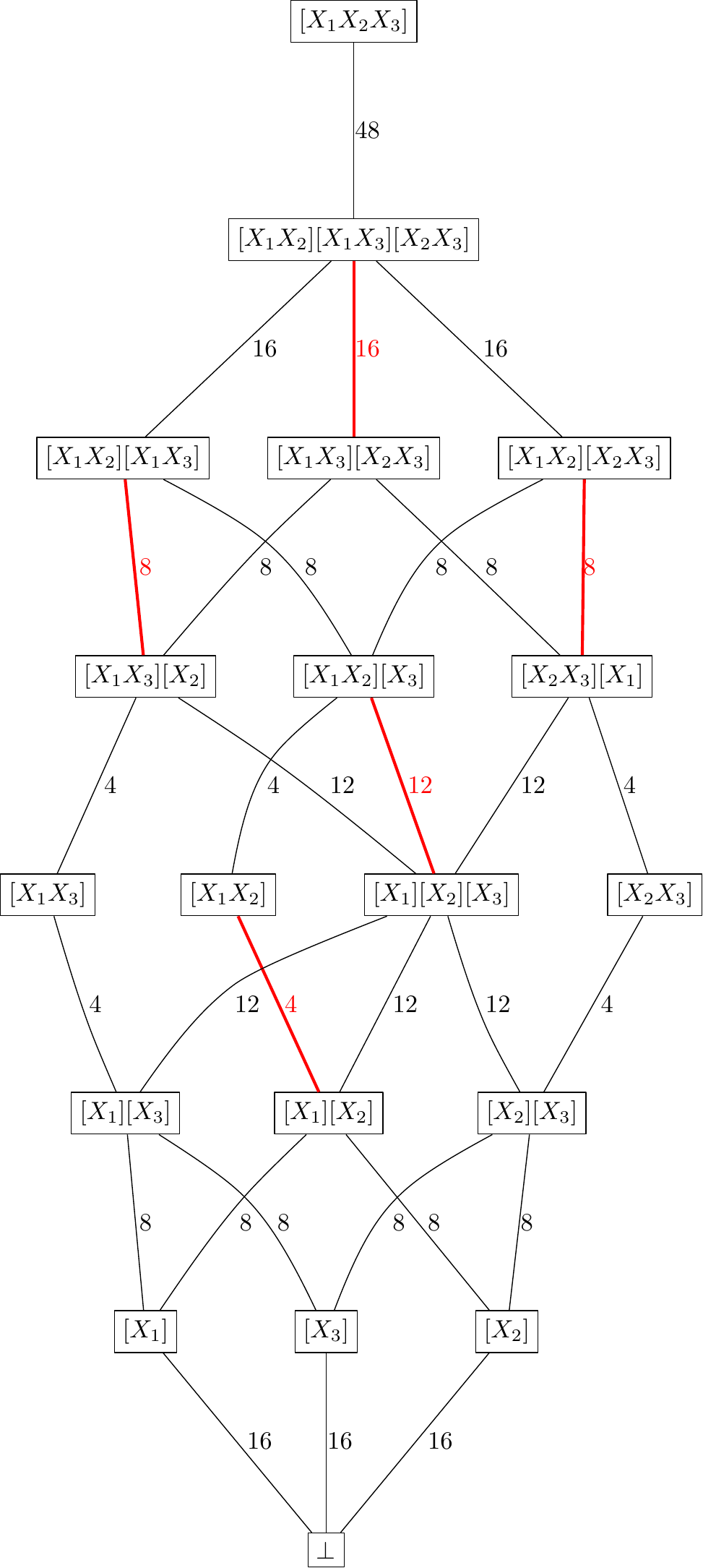}
 \caption{
 	The input lattice for three inputs. Each edge is labelled with the total number of maximal chains that pass through that that edge.  The edges where the subset $\{X_1,X_2\}$ appears for the first time are highlighted in red. Each maximal chain passes through exactly one of these edges. The contribution of $\{X_1,X_2\}$ to the total information is calculated by averaging over these edges, weighted by their path counts (the numbers in red.) {In this lattice there are 48 maximal chains in total.}
 }
 \label{lattice_3_inputs.fig}
\end{figure}

We now consider a probability weighting over the maximal chains, that is, a set of weights
$\mu(\gamma)$ such that $\sum_{\gamma\in\Gamma}\mu(\gamma)=1.$ We obtain     
\begin{align}
I(X_1,\dots,X_n; Y) 
   & = & \kl(p_{[X_1,\dots,X_n]} \, \| \, p_{\bot_\text{in}}) \\
   & = & \sum_{\gamma \in \Gamma} \mu(\gamma) \sum_{l = 1}^{2^n-1} \kl(p_{\gamma(l)} \, \| \, p_{\gamma(l - 1)})  \label{uses-normalization} \\
    & = & \sum_{\gamma \in \Gamma} \mu(\gamma) 
             \sum_{\emptyset \not= A \subseteq V} \kl(p_{\gamma(l_\gamma(A))} \, \| \, p_{\gamma(l_\gamma(A)-1)})        \\
   & = & \sum_{\emptyset \not= A \subseteq V} \sum_{\gamma \in \Gamma}  \mu(\gamma)  \kl(p_{\gamma(l_\gamma(A))} \, \| \, p_{\gamma(l_\gamma(A)-1)})         \\ 
      & = & \sum_{\emptyset \not=A \subseteq V}   \underbrace{\sum_{(\mathfrak{S}, \mathfrak{S}') \in {\mathcal E}_A} 
                \left\{ \sum_{\gamma \in \Gamma(\mathfrak{S}, \mathfrak{S}')} \mu(\gamma) \right\}  }_{=1}
             \kl( p_{\mathfrak{S}'} \, \| \, p_{\mathfrak{S}})   \label{normalization-of-mu} \\
   & = & \sum_{\emptyset \not=A \subseteq V}   \sum_{(\mathfrak{S}, \mathfrak{S}') \in {\mathcal E}_A} \mu(\mathfrak{S}, \mathfrak{S}') \, 
             \kl( p_{\mathfrak{S}'} \, \| \, p_{\mathfrak{S}})    \label{withweights}
\end{align}
Here, we used the short-hand notation $\mu(\mathfrak{S}, \mathfrak{S}')$ for 
$\mu(\Gamma(\mathfrak{S}, \mathfrak{S}')) = \sum_{\gamma \in
  \Gamma(\mathfrak{S}, \mathfrak{S}')} \mu(\gamma)$. 
  The equality (\ref{uses-normalization}) follows because of the Pythagorean theorem (\cref{full_pythagorean_eqn}) and the normalization of the weights $\mu$.
  Note, via (\ref{normalization-of-mu}), that the non-negative weights in the decomposition (\ref{withweights}) 
satisfy the following condition:
\begin{equation}
            \sum_{(\mathfrak{S}, \mathfrak{S}') \in {\mathcal E}_A} \mu(\mathfrak{S}, \mathfrak{S}') \; = \; 1\;. 
\end{equation}
This allows us to interpret 
\begin{equation}
  \label{eq:1}
       I^{(\mu)}_A(X_1,\dots,X_n\,;\,Y)
       \; := \; \sum_{(\mathfrak{S}, \mathfrak{S}') \in {\mathcal E}_A} \mu(\mathfrak{S}, \mathfrak{S}') \, 
             \kl( p_{\mathfrak{S}'} \, \| \, p_{\mathfrak{S}}) 
\end{equation}
as the mean information in $A$ that is not contained in a proper
subset of $A$. 

This gives us a non-negative decomposition of $I(X_1,\dots, X_n ; Y)$
into terms corresponding to each subset of the input variables, but
note that this decomposition is  dependent on the choice of weights $\mu$. 

A natural choice for the weights $\mu$ would be simply to choose the
uniform distribution, i.e.\ $\mu(\gamma) = 1/|\Gamma|$ for all
$\gamma$. It is not completely straightforward to justify the uniform
distribution over maximal chains, because there is no obvious symmetry
that transforms one maximal chain into another. Note, for example,  that the connectivity of the node
$[X_1][X_2][X_3]$ in
\cref{lattice_3_inputs.fig} is different from that of other nodes on the same
level.

Nevertheless, we will now proceed with the uniform distribution as 
a reasonable intuitive choice. It will be shown in  \cref{properties} that choosing $\mu$ this way gives
rise to a decomposition of $I(X_1,\dots,X_n;Y)$ that has some
intuitively desirable properties. For the special choice of
$\mu$ as the uniform distribution we will write
$I^{(\mu)}_A(X_1,\dots,X_n\,;\,Y)$ simply as
$I_A(X_1,\dots,X_n\,;\,Y)$. We denote this as the \emph{information
  contribution} of $A$ to $Y$. In
\cref{cooperative_game_theory} we will then proceed to show that above
originally merely intuitive choice of $\mu$ as uniform distribution
finds a deeper justification in the theory of cooperative game theory.

For the practical calculation of  $I_A$ we first calculate the number
$n_{(\simpcomp,\simpcomp')}$ of maximal chains that pass through each
edge $(\simpcomp,\simpcomp')$ in the Hasse diagram of the input
lattice. These numbers are shown in \cref{lattice_3_inputs.fig}, as well
as the total number of maximal chains, $|\Gamma|$.
  For each node
$\simpcomp$ in the lattice we calculate the distribution $p_\simpcomp$
by iterative scaling \citep[][chapter 5]{csiszar2004information}, from which we obtain
$D_\text{KL}(p_\simpcomp \| p_{\{\emptyset\}})$. We then find the set
$\mathcal{E}_A$ of edges in which a given predictor $A$ is added for
the first time in a maximal chain (for the example of $A=\{X_1, X_2\}$
this is shown in red in \cref{lattice_3_inputs.fig}). We then
calculate our measure $I_A$ from the Kullback-Leibler gains accrued on
these paths by adding the set $A$ of interest, weighted by the
chain counts $n_{(\simpcomp,\simpcomp')}$ of the respective edges:
\begin{equation}
 I_A(X_1,\dots,X_n\,;\,Y) = \frac{1}{|\Gamma|} \sum_{(\simpcomp,\simpcomp')\in \mathcal{E}_A} n_{(\simpcomp,\simpcomp')} \big( D_\text{KL}(p_{\simpcomp'} \| p_{\{\emptyset\}}) - D_\text{KL}(p_\simpcomp \| p_{\{\emptyset\}}) \big)\;.
 \label{our_measure}
\end{equation}

\section{Properties of the information contribution}
\label{properties}

We now prove the following properties of the information contribution, as a decomposition of the mutual information.

\begin{theorem}
\label{properties.thm}
For $A\subseteq \{X_1, \dots, X_n\}$ we have
\begin{enumerate}[I.]
 \item $I_A(X_1,\dots, X_n\,;\,Y) \ge 0$ (\emph{nonnegativity}) \label{nonnegativity.property}
 \item $\sum_{A\in 2^V} I_A(X_1,\dots, X_n\,;\,Y) = I(X_1,\dots, X_n\,;\,Y)$ (\emph{completeness}) \label{completeness.property}
 \item $I_A(X_1,\dots, X_n\,;\,Y)$ is invariant under permutations of $X_1,\dots,X_n$. (\emph{symmetry}) \label{symmetry.property}
 \item $I_A(X_1,\dots,X_n\,;\,(X_1,\dots,X_n)) = 0$ if $|A|>1$.
   (\emph{{singleton}}) \label{identity.property} 
 \item if $X_i = (X'_i, X''_i)$ for all $i$, $Y =  (Y', Y'')$, and 
 \[p(x_1, \dots, x_n, y) =  p(x'_1, \dots, x'_n, y') \, p(x''_1, \dots, x''_n, y'')\;,\]
 then 
\[
 I_A(X_1, \dots, X_n\,;\,Y) =  I_{A'}(X'_1, \dots, X'_n\,;\,Y') +  I_{A''}(X''_1, \dots, X''_n\,;\,Y'')\;,
\]
where $A' = \{X'_i : (X'_i, X''_i) \in A\}$ and $A'' = \{X''_i : (X'_i, X''_i) \in A\}$. (\emph{additivity}) \label{additivity.property}
\end{enumerate} 
As we discuss below, the singleton property is somewhat analogous to the identity axiom
proposed by \citep{harder13:_bivar} for partial information
decomposition measures, which effectively says that there should be no synergy
terms if the output is simply an indentical copy of the input.

\end{theorem}
\begin{proof}
(\ref{nonnegativity.property}) follows from the nonnegativity of the Kullback-Leibler divergence. (\ref{completeness.property}) is proved in Section~\ref{sec:lattice_decomposition} above. (\ref{symmetry.property}) is true by construction, since the values of the Kullback-Leibler divergences do not depend on the order in which the input variables are considered, and the uniform distribution over maximal chains is invariant to reordering the input variables.

To prove (\ref{identity.property}), write $Y=(\widebar X_1, \dots,
\widebar X_n)$, where $\widebar X_i$ is considered to be a copy of
$X_i$, in the sense that $X_i$ and $\widebar X_i$ are separate random
variables but we have 
\begin{equation}
\label{copy_variable}
 p(x_1,\dots, x_n, \widebar x_1, \dots, \widebar x_n) = 
\begin{cases}
 p(x_1,\dots, x_n) &\text{if $x_1=\widebar x_1, \;\dots,\; x_n=\widebar x_n\;,$} \\
 0 & \text{otherwise\;,}
\end{cases}
\end{equation}
which implies that $p(x_i, \widebar x_i) = \delta_{x_i,\widebar
  x_i} p(x_i)$, for every $i$.
We then have
\[p_{\bot_\text{in}}(X_1,\dots,X_n,\widebar X_1, \dots, \widebar X_n) =
  p(X_1)\dots p(X_n)p(\widebar X_1, \dots, \widebar X_n)\;.
\]

Consider now any edge $(\simpcomp{\setminus}\{A\}, \simpcomp)$
in the Hasse diagram of the input lattice. There are two cases to consider:
\begin{enumerate}[(i)]
 \item $A = \{X_i\}$ for some $i$. In this case $\sigma(\simpcomp)$ contains the element $\{X_i, Y\}$. Therefore, from its definition, the marginal
   $p_\simpcomp(X_i,Y)$ must match the true marginal $p(X_i,Y)$, which implies that $p_\simpcomp(x_i, \widebar x_i) =
   \delta_{x_i,\widebar x_i} p(x_i)$. 
 However, $\sigma(\simpcomp\setminus\{A\})$ does not contain the element $\{X_i, Y\}$, and so the marginal 
$p_{\simpcomp\setminus\{A\}}(x_i, \widebar x_i)$ may in general differ from $\delta_{x_i,\widebar x_i} p(x_i)$,
   and $D_\text{KL}(p_{\simpcomp}\|p_{\simpcomp\setminus A})$ can be nonzero.
 \item $|A|>1$. Consider first the case that $A = \{X_i, X_j\}$.
   Because $\simpcomp$ is a simplicial complex, we have that
   $\{X_i\}\in \simpcomp{\setminus}\{A\}$ and $\{X_j\}\in
   \simpcomp{\setminus}\{A\}$. Therefore $p_{\simpcomp\setminus
     \{A\}}$ has to match the constraints $p_{\simpcomp\setminus
     \{A\}}(x_i, \widebar x_i) = \delta_{x_i,\widebar x_i} p(x_i)$ and
   $p_{\simpcomp\setminus \{A\}}(x_j, \widebar x_j) = \delta_{x_j,\widebar
     x_j} p(x_j)$. We also have, { from \cref{copy_variable},
that $p_{\simpcomp\setminus \{A\}}( \widebar x_i, \widebar x_j) = p(
 \widebar x_i,  \widebar x_j)$. From these constraints we have
\[p_{\simpcomp\setminus A}(x_i,x_j,\widebar x_i, \widebar x_j) = p(\widebar x_i,  \widebar x_j)\delta_{x_i,\widebar x_i}  \delta_{x_j,\widebar x_j}  = p(x_i,x_j,\widebar x_i, \widebar x_j)\;.\] 
}
Therefore $p_{\simpcomp\setminus A}$ already meets the constraint that
the marginals for $X_i, X_j, Y$ match those of the true distribution
and minimising $D_\text{KL}(p_{\simpcomp}\|p_{\simpcomp\setminus A})$ subject to this constraint must result in zero. The proof of this is similar if $|A|>2$.
 \end{enumerate}
Therefore every term in \cref{withweights}
will be zero if $|A|>1$, but in general they can be nonzero if
$|A|=1$. 

To prove (\ref{additivity.property}) we first note the following
general additivity property of the Kullback-Leibler divergence. Let $Z'$ and $Z''$ be two co-distributed random variables,  
let $p_0(z',z'')=p_0(z')p_0(z'')$ for
each $z'$, $z''$ in the sample spaces of $Z'$, $Z''$, that is, render
the two random variables  independent according to the distribution
$p_0$. Then let $M$ be a mixture family defined by constraints that
depend only on either $Z'$ or  $Z''$. That is, 
\begin{align}
\label{separable_m_family}
\begin{split}
 M=\Big\{q: \sum_{z'}q(z')f^{(i)}(z')=F^{(i)}&\quad(i=1,\dots,r)\;, \\
	\sum_{z''}q(z'')g^{(j)}(z'')=G^{(j)}&\quad(j=1,\dots,s)\Big\}\;.
\end{split}
\end{align}
Calculating $\argmin_{p\in M} D_\text{KL}(p\|p_0)$, introducing Lagrange multipliers in the usual way, gives us 
\[
p(z',z'') = p_0(z')p_0(z'')e^{\sum_i \lambda_i f^{(i)}(z') + \sum_j \eta_j g^{(j)}(z'') - \psi } = p(z')p(z'')\;,
\]
where $p(z')=p_0(z')e^{\sum_i \lambda_i f^{(i)}(z') -\psi'}$ and
$p(z'') = p_0(z'')e^{\sum_j \eta g^{(j)}(z'') -\psi''}$. Note that
these are the same distributions that would be obtained if the
projection  were performed on each of the marginals rather than the joint distribution. We have both that $Z'$ and $Z''$ remain independent after projecting onto $M$, and also that $D_\text{KL}(p(Z',Z'')\,\|\,p_0(Z',Z'')) = D_\text{KL}(p(Z')\,\|\,p_0(Z'))  + D_\text{KL}(p(Z'')\,\|\,p_0(Z''))$.

{
Now consider constructing a system of random variables $X_i = (X'_i, X''_i)$, $Y =  (Y', Y'')$, according to the condition of property~\ref{additivity.property}. Each of the split distributions is defined as a projection from the product distribution onto a mixture family. By construction, all of these mixture families satisfy \cref{separable_m_family}. Because of this, every term in \cref{eq:1} can be written as a sum of the corresponding terms for the systems $\{X'_1,\dots,X'_n,Y'\}$ and $\{X''_1,\dots,X''_n,Y''\}$. The additivity property follows from this.
}

\end{proof}

We note that these properties do not uniquely determine the
information contribution measure. In particular, one could choose a
different measure $\mu$ over the maximal chains besides the uniform
measure; there are in general many such measures that would yield an
information measure satisfying \cref{properties.thm}. {To see this,
note that properties \ref{nonnegativity.property}, \ref{completeness.property},
\ref{identity.property} and \ref{additivity.property} do not depend on the choice of
measure $\mu$, and hence don't constrain it. Property \ref{symmetry.property}
does restrict the choice of measure, but for more than two inputs the number 
of paths in the lattice is greater than the number of inputs,
and consequently
the symmetry axiom does not provide enough constraint to uniquely specify
$\mu$.
}
However, as argued above, the uniform measure is
a natural choice, and we will show below that its use can be more
systematically justified from the perspective of cooperative game
theory.

\subsection{Comparison to partial information decomposition}

As noted above, the information contribution $I_A$ is not a partial
information decomposition (PID) measure, because it decomposes the
mutual information into a different number of terms than the latter. In the case of two input variables $X_1$ and $X_2$, the PID has four terms (synergy, redundancy, and two unique terms), whereas the information contribution has only three, $I_{\{X_1\}}$, $I_{\{X_2\}}$ and $I_{\{X_1, X_2\}}$. The joint term, $I_{\{X_1, X_2\}}$, behaves somewhat similarly to a synergy term, and the two singleton contributions $I_{\{X_1\}}$ and $I_{\{X_2\}}$ are roughly analogous to the two unique information terms, but there is no term corresponding to shared/redundant information. For more than two inputs, the terms of a partial information measure can be expressed in terms of a lattice known as the redundancy lattice \cite{williams2010nonnegative}, which is different from the constraint lattice or the input lattice discussed above.

Within the PID framework, \citep{harder13:_bivar} introduced the \emph{identity axiom}, which states that a measure of redundant information $I_\cap$, should satisfy
$$
 I_\cap(X_1,X_2;(X_1,X_2)) = I(X_1;X_2)\;.
$$
This is equivalent to saying that the corresponding measure of synergy, $I_\cup$, should be zero in the case where the output is a copy of its two input variables:
\begin{equation}
  I_\cup(X_1,X_2;(X_1,X_2)) = 0\;.
  \label{pid_identity}
\end{equation}
It was proven by \cite{rauh2014reconsidering} that there can be no
non-negative PID measure that satisfies all of Williams and Beer's
axioms together with the identity axiom. {This can be achieved if we restrict ourselves to two input variables, but for three or more inputs there are distributions for which it cannot be achieved. (See Example \textsc{RBOJ} below.)}

While our information contribution measure is not a PID measure, 
 if we take the joint term $I_{\{X_1,
  X_2\}}$ to be analogous to a synergy term, then the singleton decomposition property
(property~\ref{identity.property}), for two inputs, is roughly analogous to
\cref{pid_identity}. Therefore our measure obeys an analog of
the identity axiom for PID measures, alongside analogs of the
non-negativity and symmetry axioms for PID measures. This is possible
only because the information contribution is not a PID measure, and
hence does not have to obey the precise set of Williams-Beer lattice axioms.

It is also worth comparing our information decomposition measure with the framework proposed by \cite{james2017multivariate}, which seeks
a different kind of information decomposition from PID. In this framework, instead of decomposing the mutual information between
a set of sources and a target, one instead wishes to decompose the joint entropy $H(Z_1, \dots, Z_n)$ of several jointly distributed
random variables, into a sum of terms corresponding to each subset of the variables. Our framework sits somewhere between this approach
and PID, since we have the distinction between the inputs and the target, but we decompose $I(X_1,\dots,X_n;Y)$ into a sum of terms corresponding
to subsets of the inputs, in a similar manner to James and Crutchfield's proposal.

\section{Examples}
\label{examples.sec}

We now explore a few examples of our information contribution measure
(which we will also sometimes denote by Shapley information
decomposition, as will be justified by the
game-theoretic analysis in section~\ref{cooperative_game_theory}
below). Here, we apply it to joint distributions between a target and two or three
inputs. Note that our framework does not require any restrictions on
these joint distributions. In particular, it is expressly not assumed that the
inputs are independent of one another. Importantly, the measures
will in general be affected by dependent inputs,
which is a desirable property of such a measure, because it has been
observed before that appropriate attributions of joint
interactions should depend on  input correlations \citep[see the discussion
on \emph{source} vs.\ \emph{mechanistic}
redundancy in][]{harder13:_bivar}.

We take most of our examples from the literature on partial
information decomposition, in particular \citep{williams2010nonnegative,griffith2014quantifying,harder13:_bivar, bertschinger2014quantifying}. These examples are relatively standardised, and give some intuition for how our measure compares to PID measures. 

We first explore some basic examples with two predictors, which are
presented in \cref{twoinputs.tbl}. For each of these examples,
the method attributes an amount of information contribution to the predictors $\{X_1\}$, $\{X_2\}$ and $\{X_1, X_2\}$.
The numbers assigned to these sets are nonnegative and, together, they
sum up to the mutual information $I(X_1,X_2;Y)$. 
This is in many ways similar to the partial information decomposition
framework, but we note again that the information contribution decomposition
has fewer terms than the partial information decomposition (e.g.\ three rather than four in the two-input case).

\begin{table}
\centering

%
%     Rdn
%     ---
%
\begin{tabular}[t]{ccc|c}
\multicolumn{4}{c}{\textsc{Rdn}}
\\[0.4em]
$X_1$ & $X_2$ & $Y$ & $p$\\
\hline
0 & 0 & 0 & $1/2$ \\
1 & 1 & 1 & $1/2$ \\
\end{tabular}
\qquad
\begin{tabular}[t]{c|c}
	\multicolumn{2}{c}{}
	\\[0.4em]
	predictor & contribution (bits) \\
	\hline
	\\[-1em]
$\{X_2\}$ & $1/2$\\ %0.5
$\{X_1\}$ & $1/2$\\ %0.5
$\{X_1,X_2\}$ & $0$\\ %0.0
\end{tabular}
\\[1em]

%
%     Xor
%     ---
%
\begin{tabular}[t]{ccc|c}
\multicolumn{4}{c}{\textsc{Xor}}
\\[0.4em]
$X_1$ & $X_2$ & $Y$ & $p$\\
\hline
0 & 0 & 0 & $1/4$ \\
0 & 1 & 1 & $1/4$ \\
1 & 0 & 1 & $1/4$ \\
1 & 1 & 0 & $1/4$ \\
\end{tabular}
\qquad
\begin{tabular}[t]{c|c}
	\multicolumn{2}{c}{}
	\\[0.4em]
	predictor & contribution (bits) \\
	\hline
	\\[-1em]
$\{X_2\}$ & $0$\\ %0.0
$\{X_1\}$ & $0$\\ %0.0
$\{X_1,X_2\}$ & $1$\\ %0.9999999999999998
\end{tabular}
\\[1em]

%
%     2 bit copy
%     ----------
%
\begin{tabular}[t]{ccc|c}
\multicolumn{4}{c}{\textsc{2 bit copy}}
\\[0.4em]
$X_1$ & $X_2$ & $Y$ & $p$\\
\hline
0 & 0 & 0 & $1/4$ \\
0 & 1 & 1 & $1/4$ \\
1 & 0 & 2 & $1/4$ \\
1 & 1 & 3 & $1/4$ \\
\end{tabular}
\qquad
\begin{tabular}[t]{c|c}
	\multicolumn{2}{c}{}
	\\[0.4em]
	predictor & contribution (bits) \\
	\hline
	\\[-1em]
$\{X_2\}$ & $1$\\ %1.0
$\{X_1\}$ & $1$\\ %1.0
$\{X_1,X_2\}$ & $0$\\ %0.0
\end{tabular}
\\[1em]

%
%     And
%     ---
%
\begin{tabular}[t]{ccc|c}
\multicolumn{4}{c}{\textsc{And}}
\\[0.4em]
$X_1$ & $X_2$ & $Y$ & $p$\\
\hline
0 & 0 & 0 & $1/4$ \\
0 & 1 & 0 & $1/4$ \\
1 & 0 & 0 & $1/4$ \\
1 & 1 & 1 & $1/4$ \\
\end{tabular}
\qquad
\begin{tabular}[t]{c|c}
	\multicolumn{2}{c}{}
	\\[0.4em]
	predictor & contribution (bits) \\
	\hline
	\\[-1em]
$\{X_2\}$ & $0.40563765$\\ %0.40563765324637613
$\{X_1\}$ & $0.40563762$\\ %0.40563762097127865
$\{X_1,X_2\}$ & $0$\\ %1.8625270848576913e-07
\end{tabular}
\\[1em]

%
%     SynRdn
%     ------
%
\begin{tabular}[t]{ccc|c}
\multicolumn{4}{c}{\textsc{SynRdn}}
\\[0.4em]
$X_1$ & $X_2$ & $Y$ & $p$\\
\hline
0 & 0 & 0 & $1/8$ \\
0 & 1 & 1 & $1/8$ \\
1 & 0 & 1 & $1/8$ \\
1 & 1 & 0 & $1/8$ \\
2 & 2 & 2 & $1/8$ \\
2 & 3 & 3 & $1/8$ \\
3 & 2 & 3 & $1/8$ \\
3 & 3 & 2 & $1/8$ \\
\end{tabular}
\qquad
\begin{tabular}[t]{c|c}
	\multicolumn{2}{c}{}
	\\[0.4em]
	predictor & contribution (bits) \\
	\hline
	\\[-1em]
$\{X_2\}$ & $1/2$\\ %0.5000000000000001
$\{X_1\}$ & $1/2$\\ %0.5000000000000001
$\{X_1,X_2\}$ & $1$\\ %1.0000000000000002
\end{tabular}
\\[1em]

\caption{Examples of the Shapley information decomposition for several simple two-predictor cases. For each example the joint distribution is shown on the left, and on the right we tabulate $I_{\{X_1\}}$, $I_{\{X_2\}}$ and $I_{\{X_1,X_2\}}$, the contributions made by the two singleton predictors $\{X_1\}$ and $\{X_2\}$ and the joint predictor $\{X_1,X_2\}$. These three values always sum to the mutual information $I(X_1,X_2;Y)$. All logarithms are taken to base 2, so that the numbers are in bits. The interpretation of these examples is given in the text.}
\label{twoinputs.tbl}
\end{table}

In the example \textsc{Rdn} in \cref{twoinputs.tbl}, the two inputs share a single bit of information about the target. In the PID framework, this typically corresponds to one bit of shared or redundant information. However, the Shapley decomposition does not try to identify redundancy as a separate term, and instead assigns half a bit to each of the predictors. The joint predictor $\{X_1, X_2\}$ is assigned a zero contribution. This reflects the fact that once the correlations between $Y$ and the two individual predictors are known the three-way correlations are already determined, and so learning them does not reveal any extra information.

In the second example, \textsc{Xor}, we have $Y=X_1\oplus X_2$, where
$\oplus$ is the exclusive-or function. In this example, no contribution
is assigned to the individual predictors $X_1$ and $X_2$, but one bit
is assigned to the joint predictor $\{X_1,X_2\}$. This can be seen as
a kind of synergy measure --- it says that all of the information that
the predictors give about the target is found in the three-way
correlations between  $X_1$, $X_2$ and $Y$, and none in the pairwise
correlations between either predictor and the target. Interpreting
this causally, it means that the causal influences of $X_1$ and $X_2$
on $Y$ are strongly tied together. While the Shapley decomposition
does not have a term corresponding to redundancy, we see that it characterises synergy in a rather intuitive way.

Our third and fourth examples are discussed in \citep{harder13:_bivar}. The ``two bit copy'' operation plays an
important role in the PID literature in the context of the identity
axiom. The Shapley decomposition assigns one bit each to both of the
predictors and none to the joint predictor, reflecting the fact that
the two inputs each provide a different piece of information about the
target. This can be compared to the PID framework, since it is usually
seen as desirable for a PID measure to assign zero bits of synergy in
this case. Note, however, that because our decomposition does not try
to identify redundancy, it does not distinguish between this case and
the case of \textsc{Rdn}, where the information is also shared equally
between the two predictors. The results for the \textsc{And}
distribution are similar, telling us that there is also no synergy in this case. {This is because for \textsc{And} the joint distribution can be inferred completely
by knowing the marginals $(X_1,Y)$, $(X_2,Y)$ and $(X_1,X_2)$, and consequently there is no 
triplewise information.}

Our final two-predictor example is \textsc{SynRdn}, which can be formed by combining the \textsc{Xor} example with an independent copy of the \textsc{Rdn} example. The values assigned to the two predictors and the joint predictor are simply the sum of their values in the original two examples, which is a result of the additivity property (\cref{properties.thm}, part \ref{additivity.property}).

\begin{table}
\centering

%
%     Parity
%     ------
%
\begin{tabular}[t]{cccc|c}
\multicolumn{4}{c}{\textsc{Parity}}
\\[0.4em]
$X_1$ & $X_2$ & $X_3$ & $Y$ & $p$\\
\hline
0 & 0 & 0 & 0 & $1/8$ \\
0 & 0 & 1 & 1 & $1/8$ \\
0 & 1 & 0 & 1 & $1/8$ \\
0 & 1 & 1 & 0 & $1/8$ \\
1 & 0 & 0 & 1 & $1/8$ \\
1 & 0 & 1 & 0 & $1/8$ \\
1 & 1 & 0 & 0 & $1/8$ \\
1 & 1 & 1 & 1 & $1/8$ \\
\end{tabular}
\qquad
\begin{tabular}[t]{c|c}
	\multicolumn{2}{c}{}
	\\[0.4em]
	predictor & contribution (bits) \\
	\hline
	\\[-1em]
$\{X_1\}$ & $0$\\ %0.0
$\{X_2\}$ & $0$\\ %0.0
$\{X_3\}$ & $0$\\ %0.0
$\{X_1,X_2\}$ & $0$\\ %0.0
$\{X_1,X_3\}$ & $0$\\ %0.0
$\{X_2,X_3\}$ & $0$\\ %0.0
$\{X_1,X_2,X_3\}$ & $1$\\ %1.0000000000000002
\hline
total & 1 \\ % 1.0000000000000002
\end{tabular}
\\[1em]
% The fourth variable is the XOR of the other three.

%
%     XorMultiCoal
%     ------------
%
\begin{tabular}[t]{cccc|c}
\multicolumn{4}{c}{\textsc{XorMultiCoal}}
\\[0.4em]
$X_1$ & $X_2$ & $X_3$ & $Y$ & $p$\\
\hline
4 & 0 & 4 & 0 & $1/8$ \\
0 & 2 & 2 & 0 & $1/8$ \\
1 & 1 & 0 & 0 & $1/8$ \\
5 & 3 & 6 & 0 & $1/8$ \\
5 & 1 & 4 & 1 & $1/8$ \\
1 & 3 & 2 & 1 & $1/8$ \\
0 & 0 & 0 & 1 & $1/8$ \\
4 & 2 & 6 & 1 & $1/8$ \\
\end{tabular}
\qquad
\begin{tabular}[t]{c|c}
	\multicolumn{2}{c}{}
	\\[0.4em]
	predictor & contribution (bits) \\
	\hline
	\\[-1em]
$\{X_1\}$ & $0$\\ %0.0
$\{X_2\}$ & $0$\\ %0.0
$\{X_3\}$ & $0$\\ %0.0
$\{X_1,X_2\}$ & $1/3$\\ %0.3333333333333334
$\{X_1,X_3\}$ & $1/3$\\ %0.3333333333333334
$\{X_2,X_3\}$ & $1/3$\\ %0.3333333333333334
$\{X_1,X_2,X_3\}$ & $0$\\ %0.0
\hline
total & 1 \\ % 1.0000000000000002
\end{tabular}
\\[1em]
% 
% 	This is the 'XorMultiCoal' example from Griffith and Koch (2014), the paper that independently
% 	<BR/>invented the same measure as the Bertschinger et al. one. In this example, any single
% 	<BR/>predictor gives no information about the target, but any pair of predictors completely
% 	<BR/>determines it.
% 	

%
%     Bertschinger
%     ------------
%
\begin{tabular}[t]{cccc|c}
\multicolumn{4}{c}{\textsc{RBOJ}}
\\[0.4em]
$X_1$ & $X_2$ & $X_3$ & $Y$ & $p$\\
\hline
0 & 0 & 0 & 0 & $1/4$ \\
0 & 1 & 1 & 1 & $1/4$ \\
1 & 0 & 1 & 2 & $1/4$ \\
1 & 1 & 0 & 3 & $1/4$ \\
\end{tabular}
\qquad
\begin{tabular}[t]{c|c}
	\multicolumn{2}{c}{}
	\\[0.4em]
	predictor & contribution (bits) \\
	\hline
	\\[-1em]
$\{X_1\}$ & $2/3$\\ %0.6666666666666666
$\{X_2\}$ & $2/3$\\ %0.6666666666666666
$\{X_3\}$ & $2/3$\\ %0.6666666666666666
$\{X_1,X_2\}$ & $0$\\ %0.0
$\{X_1,X_3\}$ & $0$\\ %0.0
$\{X_2,X_3\}$ & $0$\\ %0.0
$\{X_1,X_2,X_3\}$ & $0$\\ %0.0
\hline
total & 2 \\ % 2.0
\end{tabular}
\\[1em]
% 
% 	This is the distribution used in Bertschinger et al.'s proof that no non-negative Williams-Beer decomposition can exist for three predictors. (No such decomposition can exist for this distribution.)
% 	<BR/><BR/>
% 	In this distribution, C = A XOR B, and then X is just a copy of the three predictors. It can be seen as a kind of 'redundant synergy', in that
% 	<BR/>
% 	A and B contain a bit of synegistic information about X that is shared redudantly with C. (Compare this to the 'synergetic redundancy' example)
% 	

%
%     three way And
%     -------------
%
\begin{tabular}[t]{cccc|c}
\multicolumn{4}{c}{\textsc{three way And}}
\\[0.4em]
$X_1$ & $X_2$ & $X_3$ & $Y$ & $p$\\
\hline
0 & 0 & 0 & 0 & $1/8$ \\
0 & 0 & 1 & 0 & $1/8$ \\
0 & 1 & 0 & 0 & $1/8$ \\
0 & 1 & 1 & 0 & $1/8$ \\
1 & 0 & 0 & 0 & $1/8$ \\
1 & 0 & 1 & 0 & $1/8$ \\
1 & 1 & 0 & 0 & $1/8$ \\
1 & 1 & 1 & 1 & $1/8$ \\
\end{tabular}
\qquad
\begin{tabular}[t]{c|c}
	\multicolumn{2}{c}{}
	\\[0.4em]
	predictor & contribution (bits) \\
	\hline
	\\[-1em]
$\{X_1\}$ & $0.18118725$\\ %0.1811872467117112
$\{X_2\}$ & $0.18118724$\\ %0.18118724468095615
$\{X_3\}$ & $0.18118724$\\ %0.18118724468095615
$\{X_1,X_2\}$ & $0$\\ %2.372686320461877e-09
$\{X_1,X_3\}$ & $0$\\ %2.3726863204583384e-09
$\{X_2,X_3\}$ & $0$\\ %1.1786918935779023e-09
$\{X_1,X_2,X_3\}$ & $0$\\ %1.2060017594167496e-07
\hline
total & 0.54356444 \\ % 0.5435644431995965
\end{tabular}
\\[1em]
% 
% 	X = A AND B AND C
% 	

\caption{Some examples of our measure, applied to joint distributions between a target and three inputs. The interpretation of these examples is given in the text.}
\label{threeinputs.tbl}
\end{table}

\Cref{threeinputs.tbl} shows the results for three input
variables. In this case the method assigns an amount of information to
every non-empty subset of $\{X_1, X_2, X_3\}$, representing the share
of the mutual information provided by that set of inputs. The first
example, \textsc{Parity}, is a three-input analog of the \textsc{Xor}
example, since $Y = X_1 \oplus X_2 \oplus X_3$. In this example it is
not possible to infer anything about the value of $Y$ until the values
of all three inputs are known. Correspondingly, the method assigns all
of the total mutual information (1 bit) to the predictor $\{X_1, X_2, X_3\}$ and none to the others.

Our second example, \textsc{XorMultiCoal} (which we take from \citep{griffith2014quantifying}) has the property that knowing any single input gives no information about the target, but any pair of predictors completely determines it. This is reflected in the contributions assigned by the Shapley decomposition: the singleton predictors $\{X_1\}$, $\{X_2\}$ and $\{X_3\}$ each make no contribution to the total. Instead, the total one bit of mutual information is shared equally between the three two-input predictors, $\{X_1, X_2\}$, $\{X_1, X_3\}$ and $\{X_2, X_3\}$. The three-input predictor $\{X_1, X_2, X_3\}$ makes no contribution, because the target is already fully determined by knowing any of the pairwise predictors.

The third example, \textsc{RBOJ}, played an important role in the
literature on PID, because it was used in
\citep{rauh2014reconsidering} to prove that no partial information
decomposition is possible that obeys the so-called identity axiom,
along with the axioms of \cite{williams2010nonnegative} and
local-positivity. In particular, no such decomposition is possible for
this distribution. In this joint distribution, the inputs $X_1$, $X_2$
and $X_3$ are related by the exclusive-or function, and the target $Y$
is in a one-to-one relationship with its inputs. As a result, each
input provides one bit (in the usual sense) of information about the target, and each pair of
inputs provides two bits, which completely determine the target.
Consequently, learning the third input adds no new information about
the target, if the other two are already known. Because our
decomposition is different from PID, it is able to assign non-negative
values to each of the predictors. It shares out the total two
bits of mutual information equally between the three singleton
predictors,  $\{X_1\}$, $\{X_2\}$~and~$\{X_3\}$. This can be seen as
a compromise between the fact  that the contributions of each member of a  pair of input variables
are independent (similarly to the 2-bit copy) and that they, at the same time, need to be fairly
allocated to three variables. 

We finish with an example, \textsc{Three way And}, in which the
decomposition is  less intuitive. In this case, the target is
1 if and only if all three inputs are 1. Similarly to the \textsc{And}
example, our measure divides the information contributions between the
three singleton predictors, assigning none to the two- or three-input
predictors. The reason for this is similar to the \textsc{And} example.
Because of this, from the perspective of our measure, this example
looks similar to the \textsc{RBOJ} example.

\section{Cooperative game theory and weighted path summation}

\label{cooperative_game_theory}

In \cref{sec:lattice_decomposition} we defined the information
contribution $I_A(X_1, \dots, X_n\,;\,Y)$ based on a uniform weighting
of the maximal chains in the input lattice. In this section we return
to the question of how this uniform distribution would be  justified.

\newcommand{\player}{A}
\newcommand{\playerset}{N}

To do so, we use the notion of the \emph{Shapley value}
\citep{Shapley1953} from cooperative game theory. Informally, the idea
of the Shapley value is that one has a set of players
$\playerset = \{\player_i, i=1\dots|\playerset|\}$. Subsets of the
players are called \emph{coalitions}, and each coalition is assigned a
\emph{{total} score} (we will use this term interchangeably with
\emph{payoff}), which is to be interpreted as how well that set of
players could do at some task, without the participation of the
remaining non-coalition players. 
Given this data, the problem is to assign a score to each individual
player, such that the scores of each individual player sum up to the
total score. The players' scores should reflect their ``fair''
contribution in achieving the total score. 

For this assignment of scores to be uniquely characterized,
Shapley postulates 
that the scores assigned to players should be a linear function of the coalitions' scores, a notion of
relevance
(explained below) and a notion of symmetry amongst the players (where
players whose contribution to value cannot be distinguished via a
symmetrical exchange of players should attain the same Shapley value).
The basic Shapley value assumes that all subsets of $\playerset$ are
possible as coalitions. Since Shapley's original work, many
generalizations of the Shapley value have been developed
\citep{Bilbao1998,Bilbao2000,Grabisch2009,Faigle2012,UlrichFaigleandMichelGrabisch2013a}. \label{sec:relat-shapl-value}

The purpose of our measure $I_A(X_1,\dots,X_n\,;\,Y)$ is to assign to each predictor $A$ a unique share of the mutual information
$I(X_1,\dots,X_n\,;\,Y)$. \Cref{our_measure} calculates this as a linear function of the quantities $D_\text{KL}(p_{\simpcomp}\|p_{\{\emptyset\}})$, which can be thought of as the information provided by $\simpcomp$, which is a set of predictors. This is  closely reminiscent to the task
of the Shapley value to identify contribution of a particular player when the values of all valid coalitions of players are known.

In fact, we can apply cooperative game theory directly to our problem,
by treating sets $A$ of input variables (i.e.\ predictors) as players in a cooperative game, in which the score of a coalition $\simpcomp$ is given by $D_\text{KL}(p_{\simpcomp}\|p_{\{\emptyset\}})$, and hence the total score is $I(X_1,\dots,X_n\,;\,Y)$. The only complication is that not every set of players forms a viable coalition, because $\simpcomp$ is constrained to be a simplicial complex. We thus need a
formulation which permits us to restrict the possible
coalitions to the ones imposed by the simplicial partial order $\leq$ on
the set $\allsimpcomp$ of all simplicial complexes. This restriction also necessitates a modification of
the symmetry axiom of the original Shapley value to guarantee that the
generalized Shapley allocation becomes uniquely determined.

Concretely, here we argue that the specific quantity in
Eq.~(\ref{eq:1}) can be interpreted precisely as the generalized
Shapley value under \emph{precedence constraints} in the sense of
\cite{Faigle1992}.

\begin{table}
  \centering
  \begin{tabular}[t]{|l|l|}
    \hline
    \textbf{Shapley Theory}& \textbf{Information in Simplicial
                             Complexes}\\
    \hline
    player $A,B,C$ & set (simplex) $A,B,C$\\
    coalition & simplicial complex\\
    
    empty coalition $\emptyset$ & empty\footnotemark\ simplicial complex
                                  $\simpcomp^{(0)} = \{\}$\\
    coalition of all players $\allplayers$ & all subsets of $\{1,\dots,
                                             n\} = [n]$, i.e.\ $2^{[n]}$ \\
    value of a coalition $\coalition$ &
                                        $D_\text{KL}(p_{\simpcomp}\|p_{\{\emptyset\}})$
                                        (0 if $\simpcomp=\{\}$)\\
    Shapley value $\phi_A(v)$ & information contribution $I_A(X_1,\dots,X_n\,;\,Y)$\\
    \hline
  \end{tabular}
  \caption{Correspondence between our quantities and coalitional game
    theory. Note that the empty coalition/simplicial complex is
    included at the bottom of the lattice.}
  \label{tab:shapley-vs-information}
\end{table}

We will use similar notation for  cooperative games to the notation we have used this far for information quantities. 
We use the symbols $A,B,C,\dots$ for players, and similarly
$\simpcomp,\simpcomp',\dots$ for coalitions, $\allsimpcomp$ for the set
of all feasible coalitions, to keep a coherent notation. Finally,
let $\allplayers$ denote the set of all players. % (i.e.\ all subsets of $[n]$).
\Cref{tab:shapley-vs-information} gives the relationship between
game-theoretic quantities and the quantities defined in previous
sections. To simplify the exposition and render it  coherent
with respect to existing literature on cooperative game theory, we additionally include the
empty coalition (simplicial complex) below the coalition
$\{\emptyset\}$ and assign to it the value 0. This will not affect any
of the results on the input lattice. 

\footnotetext{In previous chapters, we considered $\{\emptyset\}$ as
  bottom node of the input lattice. Here, we also include the node
  below $\{\emptyset\}$, interpreted as coalition, namely, the
  \emph{empty coalition} in the lattice.}

In what follows, we introduce Faigle and Kern's extension of the Shapley value, and then show that that applying it to this `information game' is indeed equivalent to \cref{eq:1} with $\mu$ taken as the uniform distribution over maximal chains, resulting in \cref{our_measure}. This demonstrates that our measure obeys the axioms of Linearity, Carrier and Hierarchical strength, described below, which are used to derive Faigle and Kern's result.

\subsection{Shapley Value under Precedence Constraints}

We now proceed to define the (generalized) Shapley value under
precedence constraints as defined in \citep{Faigle1992}. For brevity,
when we henceforth say ``Shapley value'', we will refer to this variant
unless stated otherwise.

Let $\allplayers$ be a finite partially ordered set of players, where
for $A,B \in \allplayers$ the relation $B \leq A$ enforces that, if
$A \in \coalition$ for any coalition
$\coalition \subseteq \allplayers$, one also has $B \in \coalition$
(compare with \eqref{simplcomp}). Under this constraint, not
necessarily every subset of $\allplayers$ is a valid coalition. Let $\allcoalitions$ be the set of all
valid coalitions. 
%$\coalition \subseteq \allplayers$ with this property. The
%set of coalitions 
$\allcoalitions$ 
is closed under intersection and
union operations, but not necessarily under the complement operation.

A
\emph{cooperative game} on $\allplayers$ is now a function
\begin{equation}
  \label{eq:game-value}
  v: \allcoalitions \to \mathbb{R}
\end{equation}
such that $v(\emptyset) = 0$. Consider the vector space
$\cooperativegame$ of all cooperative games on $\allplayers$, then the
Shapley value is defined as a function
\begin{equation}
  \label{eq:shapley-value-type}
  \Phi:\cooperativegame \to \mathbb{R}^\allplayers
\end{equation}
which defines, for each player $A$ from $\allplayers$, their share
$\Phi_A(v)$ for the game $v$.

Several axioms are postulated for the Shapley value.

\begin{axiom}[Linearity]
  For all $c\in\mathbb{R}, v,w \in \cooperativegame$, demand
  \begin{align*}
    \label{eq:axiom-linearity}
    \Phi(c\,v) & = c\, \Phi(v)\\
    \Phi(v+w) & = \Phi(v) + \Phi(w)
  \end{align*}
\end{axiom}

\begin{axiom}[Carrier]
  Call a coalition $\someothercoalition\in\allcoalitions$ a
  \emph{carrier} of $v\in\cooperativegame$ if
  $v(\coalition) = v(\coalition \cap \someothercoalition)$ for all
  $\coalition \in \allcoalitions$. Then, if $\someothercoalition$ is a
  carrier of $v$, we have
  \begin{equation}
    \label{eq:carrier}
    \sum_{A\in\someothercoalition} \Phi_A(v) = v(\someothercoalition)\;.
  \end{equation}
\end{axiom}
The carrier axiom needs a brief explanation. It unifies two intuitive
axioms that are sometimes used instead, the \emph{dummy axiom} (a
player that does not affect the value
(or payoff)
of any coalition attains a Shapley value of 0) and the
\emph{efficiency axiom} (the sum of the Shapley values of all players
sums up to the total payoff of the whole set of players).

The third axiom of the traditional Shapley value postulates that
players whose contribution to coalition payoffs are equivalent with
respect to a symmetric permutation will also receive the same Shapley
allocation. This axiom cannot be used in our case, because it requires
all subsets of $\allplayers$ to be feasible coalitions. To obtain a
unique characterization of the generalized Shapley value discussed
here, a stronger requirement needs to be imposed. There are several
axiom sets which are equivalent on the ordered coalition games
discussed here (see introduction of \cref{cooperative_game_theory} above).
We follow \citep{Faigle1992} in choosing the formulation via
hierarchical strength.

We need a number of definitions. Call an injective map
\begin{equation*}
  \pi:\allplayers\to\{1,2,\dots,|\allplayers|\}
\end{equation*}
a \emph{(feasible) ranking} of the players in $\allplayers$ if for all
$A,B \in\allplayers$ we have that $A < B$ (i.e.\ $A \leq B$ and
$A\not= B$) implies $\pi(A)<\pi(B)$.

The ranking $\pi$ of $\allplayers$ induces a ranking
$\pi_\coalition: \coalition\to\{1,2,\dots,|\coalition|\}$ on all
coalitions $\coalition\in\allcoalitions$ via
$\pi_\coalition(A) < \pi_\coalition(B)$ if and only if $\pi(A)<\pi(B)$
for all $A,B \in \coalition$. Note that only ordering, but not
numbering are equivalent in $\pi$ and $\pi_\coalition$.

We say that player $C\in\coalition$ is $\coalition$-maximal in the
ranking $\pi$ if $\pi_\coalition(C) = |\coalition|$ which is the same
as saying that $\pi(C) = \max_{A\in\coalition} \pi(A)$ or that there
is no player $A$ in coalition $\coalition$ with $C < A$.

We are now ready to express the concept of hierarchical strength: the
\emph{hierarchical strength} $h_\coalition(C)$ of the player $C$ in
$\coalition$ is defined as the proportion of (total) rankings
$\pi$ in which $C$ is $\coalition$-maximal. Formally,
\begin{equation}
  \label{eq:hierarchical-strength}
  h_\coalition(C) := \frac{1}{|\allrankings(\allplayers)|}
  |\{\pi\in\allrankings(\allplayers)\mid \text{$C$ is
    $\coalition$-maximal for $\pi$}\}|
\end{equation}
where $\allrankings(\allplayers)$ is the set of all rankings for the
set $\allplayers$ of players.

Define now a particular fundamental game type, the
\emph{inclusion game} over $\coalition$, called $\zeta_\coalition$
via:
\begin{equation}
  \label{eq:inclusion-game}
  \zeta_\coalition(\yetanothercoalition) :=
  \begin{cases}
    1 & \text{if $\coalition\subseteq \yetanothercoalition$}\\
    0 & \text{otherwise}\;.
  \end{cases}
\end{equation}
In other words, the payoff of the game is 1 if the tested coalition
$\yetanothercoalition$ encompasses a given reference coalition
$\coalition$ and vanishes otherwise. We mention without proof that
these games form a basis of $\cooperativegame$ and thus it is
sufficient to define the Shapley value over all inclusion games on
$\allplayers$. Inclusion games form a preferred set of coalitions in
the study of coalitional games since they are closely related to
elementary games (games where only a specific coalition achieves a
non-zero payoff) and have an intuitive interpretation.

Finally, we can now define 
\begin{axiom}[Hierarchical Strength (Equivalence)]
  For any $\coalition\in\allcoalitions$, $A,B\in\coalition$, we
  demand:
  \begin{equation}
    \label{eq:hierarchical-strength-axiom}
    h_\coalition(A)\Phi_B(\zeta_\coalition) = h_\coalition(B) \Phi_A(\zeta_\coalition)
  \end{equation}
\end{axiom}
Informally, the Shapley value of a player $B$ in a coalition
$\coalition$ for the inclusion game is weighted against that of
another player $A$ in the same coalition via their hierarchical
strength. All else being equal, the Shapley values of the two players
relate to each other as their hierarchical strengths --- a larger
value of the hierarchical strength corresponds to a larger Shapley
value, i.e.\ larger allocation of payoff.

\citet{Faigle1992} note that the hierarchical strength 
emphasizes the given player being \emph{on top} of its respective coalition
in the given ranking rather than, say, considering its average rank. This
is insofar an intuitive choice for the generalized Shapley
value, since it is only the top-ranked player in a coalition which
determines whether that particular coalition is formed at all. In
other words, it is a measuring in how many rankings (relative to the
total number of rankings) that particular player has the power to
decide whether the given coalition will be formed or not.

It turns out this has a straightforward reinterpretation and
generalization in the context of Markovian coalition processes
\citep{UlrichFaigleandMichelGrabisch2012}
In addition, there are many other formulations
equivalent with it (see references mentioned in the introduction to
the present chapter
\cref{cooperative_game_theory}). We opted for the formulation via
hierarchical strength since it is the most widely established
generalization of the symmetry axiom for the classical Shapley value
in the literature.

We state here without proof that the unique payoff allocation is given
by

\begin{equation}
  \label{eq:generalized-shapley-value}
  \Phi_C(v) = \frac{1}{|\allrankings(\allplayers)|}
  \sum_{\substack{\yetanothercoalition\in\allcoalitions\\
    C\in\yetanothercoalition^{\mathmakebox[0cm][l]{+}}}}
|\allrankings(\yetanothercoalition\setminus \{C\})| \cdot
|\allrankings(\allplayers\setminus\yetanothercoalition)| 
\;
\bigl(
v(\yetanothercoalition) - v(\yetanothercoalition\setminus \{C\})
\bigr)
\end{equation}

where we trivially assume $|\allrankings(\emptyset)| = 1$ and where
$C \in \yetanothercoalition^+$ sums over all coalitions
$\yetanothercoalition$ for which $C$ is maximal. The (generalized)
Shapley value of $C$ is given by the marginal contribution of $C$ to
all coalitions $\yetanothercoalition$ for which it is maximal, weighted
by the proportion of rankings for which this is the case.

We now show that our definition of information contribution of a
simplex (\cref{our_measure}) is equivalent to the generalized Shapley value
under precedence constraints if the value is the mutual information
between input variables $X_1,\dots,X_n$ and output variable $Y$. Thus,
the information contribution has a natural interpretation in the context of game-theoretic
payoff allocation.

\subsection{Equivalence of Generalized Shapley Value and the Sum over Maximal Chains}

\label{sec:equiv-gener-shapl}

\begin{theorem}
\label{theorem:equivalent-path-counting-shapley}

We now prove that, under the identifications of \cref{tab:shapley-vs-information},
the information contribution of a
set $A$ is identical with its Shapley value under precedence
constraints, with $A$ interpreted as a player. More precisely:
\begin{multline}
  \label{eq:equivalence-info-contribution-shapley}
  \sum_{(\simpcomp,\simpcomp')\in \mathcal{E}_A}
  \mu(\simpcomp,\simpcomp')\,  D_\text{KL}(p_{\simpcomp'}\,\|\,p_\simpcomp)  =\\
  \sum_{\substack{\coalition\in\allcoalitions\\ A \in \coalition^{\mathmakebox[0cm][l]{+}}}}
  \frac{|\allrankings(\coalition \setminus \{A\})| \cdot |\allrankings(2^{[n]}\setminus\coalition)|}{|\allrankings(2^{[n]})|}
  \left[
  D_\text{KL}(p_\simpcomp||p_{\{\emptyset\}}) - D_\text{KL}(p_{\simpcomp\setminus \{A\}}||p_{\{\emptyset\}})
  \right]\,,
\end{multline}
where the weighting of the lattice chains $\mu$ is chosen as the
uniform distribution, $\mu(\gamma) = 1/|\Gamma|$. To permit consistency
between lattice and Shapley model, we furthermore define the bracketed term on
the right side to be 0 for $\simpcomp=\emptyset$.

\end{theorem}

\begin{proof}
  Consider $\allsubsets = 2^{[n]}$. Identify the elements
  $A \in \allsubsets$, i.e.\ the subsets of $[n]$, with the players in a
  Shapley coalition game with partial ordering defined via the subset
  relation, i.e.\ via
  \begin{align*}
    B \leq A :\Longleftrightarrow B\subseteq A\,.
  \end{align*}
  Per definition, the partial order-compatible coalitions $\coalition$
  then precisely constitute the simplicial complexes of $\allsubsets$.

  We now show that, under these identifications, the (feasible)
  rankings of players define precisely the maximal chains over
  simplicial complexes. In other words, there is a one-to-one
  correspondence between the rankings of the ordered coalition game
  and the maximal chains over its corresponding simplicial complexes.
  
  The intuition of the proof is as follows: in the Hasse diagram for
  the input lattice \cref{sec:lattice_decomposition}, each maximal
  chain is formed by successively adding each of the predictors, one
  at a time, in such a way that each step of the chain remains a
  simplicial complex. We will demonstrate that the orders in which the predictors are added in
  a maximal chain correspond precisely to the feasible rankings of the
  predictors interpreted as Shapley
  players.   We now proceed to show this formally.

  We first show well-definedness, i.e.\  that each ranking defines a maximal chain. Let $\pi$, a
  (feasible) ranking over the set of players $\allplayers$, be given
  (we remind that each player is a subset of $[n]$). Define the sequence

  \begin{align}
    \label{eq:sequence-of-coalitions}
    \coalition_0 &:= \emptyset, \coalition_1, \coalition_2, \dots,\coalition_{|\allsubsets|}\\
    \intertext{where for $k=1,\dots,|\allsubsets|$}
    \coalition_k &:= \{A \in \allsubsets\mid \pi(A) \leq k\}\\
    & = \pi^{-1}(\{1,\dots,k\})\,. \label{eq:2}
  \end{align}

  We now need to show now that this sequence
  $(\simpcomp_k)_{k=0,\dots,|\allsubsets|}$ is, first, a chain of
  simplicial complexes (equivalently, feasible coalitions) and,
  second, maximal.

  If $k=0$, then $\simpcomp_k = \emptyset$ is trivially a simplicial
  complex. Else, let $1\leq k \leq |\allsubsets|$. Consider now
  $A\in \coalition_k$, and any $B\in\allsubsets$ with $B\subseteq A$.
  We have $\pi(B) \leq \pi(A) \in \{1,\dots,k\}$ per ranking property,
  and thus $\pi(B)\in\{1,\dots,k\}$, and thus $B\in\coalition_k$ and
  $\coalition_k$ is a simplicial complex.

  From \eqref{eq:2} it follows that, for $k\leq l$,
  $\coalition_k \subseteq \coalition_l$. Therefore, if
  $A\in\simpcomp_k$, also $A\in\simpcomp_l$ and thus $\simpcomp_k \leq
  \simpcomp_l$ and the
  $(\simpcomp_k)_k$ form a chain.

  This chain is maximal. To show this, consider successive simplicial
  complexes $\simpcomp_k,\simpcomp_{k+1}$, $k=0,\dots,|\allplayers|-1$
  in the sequence. Consider $\widetilde{\simpcomp}$ such that
  $\simpcomp_k \leq \widetilde{\simpcomp} \leq \simpcomp_{k+1}$
  according to the natural partial order $\leq$ on simplicial
  complexes. If $\simpcomp_k \not= \widetilde{\simpcomp}$, then there
  exists a $B\in \widetilde{\simpcomp} \setminus \simpcomp_k$ and,
  since $\widetilde{\simpcomp} \leq \simpcomp_{k+1}$, one has
  $B\subseteq C$ for some $C\in \simpcomp_{k+1}$. This means that
  $\pi(B) \leq \pi(C)$. Since $B \notin \simpcomp_k$, also
  $\pi(B) \notin \{1,\dots,k\}$, so, per construction of
  $\simpcomp_{k+1}$, necessarily $\pi(B)=k+1$ and
  $B= \pi^{-1}(k+1) = C \in \simpcomp_{k+1}$. It follows that
  $\widetilde{\simpcomp}$ must be either $\simpcomp_k$ or
  $\simpcomp_{k+1}$, thus, $\simpcomp_k \prec \simpcomp_{k+1}$ and the
  chain is maximal. This shows that the mapping from rankings to maximal chains is
  well-defined.
  
  We show now that mapping rankings to maximal chains
  \eqref{eq:sequence-of-coalitions} via \eqref{eq:2} is injective.
  For this, consider two rankings $\pi\not=\rho$. We have to show that
  they induce different maximal chains.

  Consider $B$ with $\pi(B)\not=\rho(B)$. Assume, without loss of
  generality, $\pi(B) < \rho(B)$. If we consider the chain
  $(\simpcomp^\pi_k)_k$ induced by $\pi$ (and analogously
  $(\simpcomp^\rho_k)_k$ for $\rho$), then observe that the chain can
  be written in the form of inclusion chain as
  \begin{align}
    \label{eq:pi-induced-sequence}
    \emptyset \subseteq \simpcomp_0^\pi \subseteq \simpcomp_1^\pi \subseteq \dots \subseteq
    \underset{\uparrow\mathmakebox[0cm][l]{\raisebox{-1ex}{\scriptsize\text{first time where $B$
    appears in $(\simpcomp^\pi_k)_k$}}}}{\simpcomp^\pi_{\pi(B)}} \subseteq \dots \subseteq
    \simpcomp^\pi_{|\allplayers|} = \allplayers\,.
  \end{align}
  In this chain, the first simplicial complex to contain $B$ is the
  one with index $\pi(B)$.  Under the same consideration for the chain
  induced by $\rho$, the first member of the chain to contain $B$ is the one with
  index $\rho(B)$. However, $\pi(B)<\rho(B)$ and therefore the chains
  must differ and assigning chains to rankings via \eqref{eq:sequence-of-coalitions} is injective.

  Show now surjectivity: for each maximal chain, there is a ranking
  that produces it.
  Let
  \begin{align}
    \label{eq:given-chain}
    \emptyset = \simpcomp_0 \subseteq \simpcomp_1 \subseteq \dots
    \subseteq \simpcomp_{|\allplayers|} = \allplayers
  \end{align}
  be a maximal chain. We show now that each step adds exactly one
  $C\in\allplayers$. Assume none of the steps in the sequence is
  trivial, i.e.\ we always have $\simpcomp_j \subsetneq \simpcomp_{j+1}$.
  All $\simpcomp_k$ are at the same time simplicial complexes as well
  as --- equivalently --- order-compatible coalitions. Choose
  $C\in\simpcomp_{j+1}\setminus\simpcomp_j$ minimal (i.e.\ such that
  for any $B\in\simpcomp_{j+1}\setminus\simpcomp_j$ with
  $B\subseteq C$, we have $B=C$).

  Since $C\in\simpcomp_{j+1}$, for any $B\subseteq C$, we have
  $B\in \simpcomp_{j+1}$. It follows that either $B\in\simpcomp_j$ or
  $B\in\simpcomp_{j+1}\setminus \simpcomp_j$; in the latter case, however, because of
  minimality of $C$ in $\simpcomp_{j+1}\setminus\simpcomp_j$, it
  follows $B = C$. Thus $\simpcomp_j \cup \{C\}$ is a simplicial
  complex, and because of maximality of the chain, it must be
  identical to $\simpcomp_{j+1}$. In summary, in each step of the
  maximal chain precisely one simplicial complex is added.

  Finally, given a maximal chain
  \begin{align}
    \label{eq:3}
    \emptyset \prec \simpcomp_1 \prec \simpcomp_2 \prec \dots \prec
    \simpcomp_{|\allplayers|} = \allplayers\,,
  \end{align}
  define for every $j=1,\dots,|\allplayers|$ the inverse ranking
  $\pi^{-1}(j)$ to map onto the unique set (player) in
  $\simpcomp_j\setminus\simpcomp_{j-1}$. The maximal chain
  \eqref{eq:3} is induced by the ranking $\pi$; we have thus shown the
  mapping \eqref{eq:2} of rankings to maximal chains to be surjective
  (for every maximal chain there is a ranking that is mapped to it).
   With the injectivity shown earlier, this mapping is thus bijective.
  In short, we have shown that to each maximal chain corresponds one
  and only one feasible ranking.

  Consider now ${(\simpcomp,\simpcomp')\in \mathcal{E}_A}$, i.e.\ an
  edge where $\simpcomp' = \simpcomp\cup\{A\}$ is obtained by adding
  $A$ to $\simpcomp$. For the set $\Gamma(\simpcomp,\simpcomp')$ of
  (maximal) chains $(\simpcomp_k)_k$ who pass through this edge, i.e.\
  for which $\simpcomp_j = \simpcomp$ and
  $\simpcomp_{j+1} = \simpcomp'$ for some $j$, one has, in analogy to the
  derivation above, a one-to-one map to the pairs of the rankings over
  $\simpcomp'\setminus\{A\}$ and those over
  $\allplayers\setminus\simpcomp'$:
  \begin{align}
    \allrankings(\simpcomp'\setminus\{A\})
    \times \allrankings(\allplayers\setminus\simpcomp')\,.
  \end{align}
  This is seen by replacing the full ranking with two subrankings, one over the
  lower sublattice with $\simpcomp$ as top element instead of
  $\allplayers$ and one over the upper one
  which has
  $\simpcomp'$ as bottom element replacing $\simpcomp_0$.
  It follows that we have
  \begin{align}
    \label{eq:chains-with-edge-vs-ranking-counts}
    |\Gamma(\simpcomp,\simpcomp')| =
    |\allrankings(\simpcomp'\setminus\{A\})| \cdot
    |\allrankings(\allsubsets\setminus\simpcomp')| \,.
  \end{align}  

  Consider a particular edge
  ${(\simpcomp,\simpcomp')\in \mathcal{E}_A}$. We note that this edge
  corresponds precisely to the simplicial complexes
  $\simpcomp'\in\allsimpcomp$ where $A$ is maximal in
  $\simpcomp'$, i.e.\ $A\in\coalition'^+$. We had earlier the short-hand
  notation
  $\mu(\mathfrak{S}, \mathfrak{S}') = \sum_{\gamma \in
    \Gamma(\mathfrak{S}, \mathfrak{S}')} \mu(\gamma)$ where
  $\Gamma(\simpcomp,\simpcomp')$ ranges over all chains containing a
  particular edge. If all chains/paths $\gamma$
  are equally weighted, their weight is given by
  \begin{align}
    \label{eq:path-weight}
    \frac{1}{|\Gamma|} =
    \frac{1}{|\allrankings(\allplayers)|} = \frac{1}{|\allrankings(2^{[n]})|}
  \end{align}

  Finally, note that 
  \begin{align}
    \label{eq:successive-kl-steps}
    D_\text{KL}(p_{\simpcomp'}\,\|\,p_\simpcomp) = D_\text{KL}(p_{\simpcomp'}||p_{\simpcomp^{(0)}}) - D_\text{KL}(p_{\simpcomp'\setminus \{A\}}||p_{\simpcomp^{(0)}})
  \end{align}
  because of the Pythagorean relation \eqref{full_pythagorean_eqn}.
  This completes the proof of \eqref{eq:equivalence-info-contribution-shapley}.

\end{proof}

Note that, when constructing the correspondence between the input
lattice to the Shapley value, for the former we had the maximal chains
start at $\{\emptyset\}$ rather than at $\emptyset$ as bottom of the
lattice. However, the property from
\cref{theorem:equivalent-path-counting-shapley} continues to hold in
this case, since the bottom step from $\emptyset$ to $\{\emptyset\}$
is unique and does not affect the path counts.

\section{Discussion}

In the search for a partial information measure that allocates
informational contributions to various input variable sets (i.e.\
predictors) we relinquished the demand to quantify redundancy and
instead applied the Pythagorean decomposition to characterize the
additional contribution of an input variable set as it is added on the
relevant maximal chains. This ``longitudinal'' contribution is
chain-dependent, though. To be able to talk about a contribution of an
individual predictor, though, we need to express this contribution
independently of the particular chain.

Intuitively, this can be done by assigning a probability distribution
over the chains and averaging a predictor's contribution over all
these chains; most naturally, the equidistribution could be chosen for
this purpose. A more justified reasoning for this choice can be
derived by observing that the setup of information contribution
precisely matches the situation of a coalition game where the value of
a coalition is the contribution of that coalition to the overall
``value'', i.e.\ information about the target variable; and that
contribution can be fairly assigned via the Shapley value concept. Of
course, with the natural precedence order of predictors, not all
coalitions of predictors (i.e.\ players in the language of game
theory) are viable. We needed to resort to the variant of the Shapley
value under \emph{precedence constraints} which, as it turns out,
corresponds precisely to the averaging over all maximal chains of the
input lattice, strengthening both the confidence in the
appropriateness of the measure and the intuition behind it.

While the view of a predictor contribution stemming from averaging
over chains (paths) through the lattice seems abstract and artificial,
the Shapley value-based interpretation justifies its use. In fact,
this perspective finds, again, additional justification from more
recent coalition game theory in which coalitions are not considered as
immutable, but can change as per a stochastic process via local
incentives \citep{UlrichFaigleandMichelGrabisch2012}. In our context,
this would correspond to a dynamically chosen path in an input
lattice. At this stage, however, we are interested in the static
contributions of the predictors; whether there will be an incentive to
invoke a complex trajectory in the input lattice over which the
contributions will be averaged, remains a question for the future.

\section*{Acknowledgement}  
DP would like to acknowledge support by H2020-641321 socSMCs FET Proactive project. NA and NV acknowledge the
support of the Deutsche Forschungsgemeinschaft Priority Programme ``The Active Self'' (SPP 2134).

\bibliography{Inf_Decomp} % bib-math

\end{document}